\documentclass[11pt,a4paper]{article}

\usepackage{float}
\usepackage{fullpage}
\usepackage{booktabs,wrapfig} 
\usepackage[ruled]{algorithm2e} 

\SetAlFnt{\small}
\SetAlCapFnt{\small}
\SetAlCapNameFnt{\small}
\SetAlCapHSkip{0pt}
\IncMargin{-\parindent}
\usepackage[numbers]{natbib} 
\usepackage{amsthm}
\usepackage{amsfonts}
\usepackage[inline]{enumitem}
\usepackage{mathtools}
\usepackage{hyperref}
\usepackage[svgnames]{xcolor}
\hypersetup{colorlinks={true},urlcolor={blue},linkcolor={DarkBlue},citecolor=[named]{DarkGreen}}

\usepackage{cleveref}
\usepackage{mathtools}
\usepackage{subcaption}

\makeatletter
\renewcommand{\@biblabel}[1]{[#1]\hfill}
\makeatother

\newcommand{\defeq}{\vcentcolon=}  
\DeclarePairedDelimiter\abs{\lvert}{\rvert}%

\newtheorem{theorem}{Theorem}
\newtheorem{definition}{Definition}
\newtheorem{lemma}{Lemma}
\newtheorem{property}{Property}

\crefname{property}{property}{properties}
\Crefname{Property}{Property}{Properties}

\newcommand{\la}{\lambda}

\title{MEV Sharing with Dynamic Extraction Rates} 

\author{
\begin{tabular}{c c c}
& & \\ 
\textbf{Pedro Braga} & \textbf{Georgios Chionas} & \textbf{Piotr Krysta\thanks{P. Krysta is also affiliated with University of Liverpool, Liverpool, UK}} \\
\small{King's College London} & \small{University of Liverpool} & \small{Augusta University}  \\
\href{mailto:pedro.braga@kcl.ac.uk}{\small{\texttt{pedro.braga@kcl.ac.uk}}} & \href{mailto:g.chionas@liverpool.ac.uk}{\small{\texttt{g.chionas@liverpool.ac.uk}}} & \href{mailto:pkrysta@augusta.edu}{\small{\texttt{pkrysta@augusta.edu}}} \\
& & \\
\textbf{Stefanos Leonardos} & \textbf{Georgios Piliouras} & \textbf{Carmine Ventre} \\
\small{King’s College London} & \small{Singapore University of Technology and Design} & \small{King’s College London}\\
\href{mailto:stefanos.leonardos@kcl.ac.uk}{\small{\texttt{stefanos.leonardos@kcl.ac.uk}}} & \href{mailto:georgios@sutd.edu.sg}{\small{\texttt{georgios@sutd.edu.sg}}} 
& \href{mailto:carmine.ventre@kcl.ac.uk}{\small{\texttt{carmine.ventre@kcl.ac.uk}}}
\end{tabular}}

\date{}

\begin{document}

\maketitle

\begin{abstract}
Maximal Extractable Value (MEV) has emerged as a new frontier in the design of blockchain systems.  
In this paper, we propose making the MEV extraction rate as part of the protocol design space. Our aim is to leverage this parameter to maintain a healthy balance between block producers (who need to be compensated) and users (who need to feel encouraged 
to transact). We follow the approach introduced by EIP-1559 and design a similar mechanism to dynamically update the MEV extraction rate with the goal of stabilizing it at a target value. We study the properties of this dynamic mechanism and show that, while convergence to the target can be guaranteed for certain parameters, instability, and even chaos, can occur in other cases. Despite these complexities, under general conditions, the system 
concentrates in a neighborhood of the target equilibrium implying high long-term performance. Our work establishes, the first to our knowledge, dynamic framework for the integral problem of MEV sharing between extractors and users.
\end{abstract}

\section{Introduction}
Smart contract blockchains, such as Ether\-eum \cite{Buterin13}, Cardano \cite{Kiayias17} and Tezos \cite{Goodman14}, unveiled the ability to construct flexible financial products on top of blockchains, which, in turn, led to the emergence of \emph{Decentralized Finance} (DeFi). DeFi markets supports most applications available in traditional finance such as asset exchanges, trading of tokens through Decentralized Exchanges (DEXes), loans, decentralized governance, and stablecoins (currency boards) \cite{Werner22}. In the absence of centralised authorities, DeFi exchanges operate as two-sided markets, with users on the one side submitting transactions and block producers on the other side executing these transactions by broadcasting them into blocks through decentralised networks. \par
However, block producers, e.g., miners in Proof of Work (PoW), validators in Proof of Stake (PoS), or block producers (and searchers) in Ethereum's proposer-builder-separation (PBS) framework have unilateral power to decide which transactions to execute and in which order \cite{Hub21,nisan2023serial,wu2024strategic}. This unfolds a rich strategy space to extract additional profit, usually at the expense of users.
This encompasses any type of excess profit a block producer can obtain by adding, ordering, or censoring transactions and is described as \emph{Maximal Extractable Value} (MEV) \cite{Daian19}. The multiple ways in MEV involve frontrunning, arbitrage and sandwich attacks in DEXes \cite{ethereum_mev, Fer22b}, liquidations \cite{Kao20,Qin21b}, cross-chain MEV \cite{Obadia21} and NFT-mint front running \cite{ethereum_mev}. While many of these practices have been prevented or regulated in traditional finance (e.g., frontrunning is known to breach certain MiFID II standards, such as, best execution ---~\href{https://www.fca.org.uk/markets/regulation-markets-financial-instruments}{fca.org}), 
the lack of centralized authorities renders such mechanisms of little or no practical relevance for DeFi markets. \par
In current DeFi markets, MEV constitutes an intrinsic and dominant feature that is an integral source of concern. The reason is that even though MEV is produced by the economic activity of users (MEV generators), in most cases, its value is shared among other actors including the block producers of DeFi markets (MEV extractors). 
Resolving this tension is far from straightforward. Extreme solutions, in which all MEV goes either back to users or to block producers, are shown to adversely affect the behaviour of the opposite group and end up compromising participation and system security \cite{Qin21,Daian19,Chitra22}.\par
Hence, the question of how to share the generated MEV between these two groups, i.e., users and block producers, constitutes a central problem in current blockchain research and practice \cite{Flashbots_suave,Flashbots_share}. The current debate around MEV \emph{sharing} or mitigation involves two main schools of thought. The first is the \emph{anti-MEV} approach, which states that MEV is harmful since it increases network congestion, higher costs for users, information asymmetries and thus it must be prevented or mitigated by adding additional ordering constraints to the consensus layer \cite{Kelkar20,Kelkar21a,Kelkar21b}. The second is the \emph{pro-MEV} approach which argues that, despite its negative externalities, MEV is endemic in permissionless blockchains, and it can present some benefits for market participant such as balanced liquidity pools and in some cases, improvement of the quality of executing transactions \cite{Barnabe22,Kulkarni23,Flashbots_suave,Chitra22}. However, despite the intense interest of the blockchain community to resolve this problem, to the best of our knowledge, a rigorous proposal with provable guarantees remains elusive. \par

\paragraph*{Our Approach: extraction rate oracles for dynamic MEV sharing}  In this work, we take a neutral approach towards MEV and propose a mechanism to balance the two extreme positions by enshrining MEV sharing in the protocol design. Specifically, we seek to incentivise (and maintain) participation of both users and block producers by stipulating a distribution of the generated MEV value among them. The key element to achieve this is a variable MEV \emph{extraction rate} that determines the fraction of MEV that is retained by block producers (or miners),\footnote{Due to its common usage, we will also use the term \emph{block producers} or \emph{miners} for simplicity as a shorthand of all parties that are involved in the process of \emph{block creation} and MEV extraction, such as validators, block proposers, builders, searchers, etc. \cite{Mon22}.} with the rest going back to users (see, e.g., \cite{Flashbots_share}). For example, an MEV extraction rate of $1$ ($0$, resp.) means that all MEV goes to miners (users, resp.).\par
The implementation of a known, yet possibly variable to reflect market conditions, MEV extraction rate in a DeFi market serves a twofold purpose. On the one hand, it ensures an improved user experience by providing an oracle to users who know what extraction to anticipate. On the other hand, it regulates the incentives provided to each side of the market in order to maintain sufficient participation by each one. While the first goal is straightforward, the second is more involved as it needs to avoid impractical optima (i.e., market configurations with very imbalanced incentives and, thereby, participation). For instance, reducing the MEV extraction to zero will myopically maximise incentives for users, and thereby, throughput, but will force the market to collapse by removing (most) incentives for block producers.\par
To resolve this tension in practice (assuming sufficient demand to exclude pathological cases, as is the case currently), the main signal that such a mechanism needs to consider from the market is the \emph{participation ratio} between users and block producers (measured appropriately, e.g., transactions by users and number of active builders in PBS). Thus, the main goal of the MEV extraction rate is to strike (and maintain) a \emph{balance} between users and miners to a predetermined target ratio, stipulated by safety or scalability considerations. Accordingly, to reflect changing market conditions, the extraction rate is updated after every block (or in regular intervals, e.g., every epoch) following a dynamic update rule. The intensity of the updates is regulated by an \emph{adjustment parameter} that is part of the design space of the mechanism. In simple terms, the rule increases (decreases, respectively) the MEV extraction rate if the ratio of miner/user participation is lower (higher, respectively) than desired.\footnote{To measure deviations from the target, we consider the winning bid of the auction held for the creation of the last block as a proxy of the extracted MEV. While the estimation of MEV is not \emph{the} focus of our model, current research suggests that winning bids in MEV-boost auctions (\cite{Flashbots_boost,Flashbots_explore}) or experimental methods (\cite{Babel21,Qin21,Bartoletti22, Heimbach22}) provide accurate estimates or tight lower bounds to the actual MEV.} \par
While the trade-off between the goals of bidders (here users) and auctioneers (here block producers) has been studied in the context of conventional auctions \cite{Likhodedov04,Diakonikolas12}, the above dynamic update rule is inspired by the design principles of EIP-1559 which has been recently implemented to regulate transaction fee markets under similar conditions \cite{Buterin19a, Roughgarden20, Rou21, Leonardos21,Leonardos2022}. The common characteristic of EIP-1559 and the currently proposed MEV extraction rule is the inclusion of the deviations from the target function in the updates. As shown in \cite{Diamandis23dynamic}, this is equivalent to solving an economically motivated constrained optimization problem by implicitly applying gradient descent on its dual.


\paragraph*{Our Results}
Our objective in this paper is to study the evolution and performance of the above mechanism. In particular, to contribute to the ongoing debate about MEV, our motivating questions are the following: 
\begin{enumerate}
    \item can a protocol-determined, dynamically-updated MEV extraction rate balance the conflicting participation incentives of users and block producers and 
    \item can the design principles of EIP-1559 be useful in accomplishing this task?
\end{enumerate} \par 
Our findings provide significant evidence that we can answer both questions affirmatively. 
First, we formally analyse the evolution of the dynamic mechanism and show that the system exhibits one of the following behaviors depending on the intensity of the updates, as these are determined by the value of the adjustment parameter: (i) \emph{low} intensities: convergence to an optimal MEV extraction rate, which is strictly between 0 and 1 that achieves the target participation ratio, (ii) \emph{intermediate} intensities: periodic behavior or provable chaos (\Cref{thm:chaotic}) or (iii) \emph{large} intensities: collapse to extraction rates of either $0$ or $1$ in which case one of the two groups of this two-sided economy abandons the system. In \Cref{thm:market_liveness} and \Cref{thm:convergence}, we establish formal threshold bounds on the intensity of the updates for which the system either avoids the degenerate states of $0$ and $1$ and remains \emph{live}, i.e., avoids regime (iii), or equilibrates (stabilizes) at the target MEV extraction rate, i.e., reaches regime (i), respectively.\par
Turning to system performance, we mathematically quantify deviations from target participation levels and show that they remain bounded for a broad range of realistic market conditions and parameterizations (\Cref{thm:performance}). This means that the system maintains a healthy proportion between miners and users. The main takeaway is that this desirable behavior is formally exhibited not only when the system equilibrates, a condition that may be rarely met in practice, but, also, when it operates under \emph{out-of-equilibrium (chaotic) conditions}.\par
\paragraph*{Technical overview and robustness of findings}
For the rigorous mathematical derivations in the above results, we consider the expected behavior, i.e., deterministic version, of the system where randomness is resolved through its expectation. Our extensive simulations (cf. \Cref{app:appendix}) demonstrate the robustness of the above system behaviour even when we lift such assumptions (which are only necessary for the mathematical proofs), i.e., when we perturb the update rule to consider closely related designs including MEV-burn (\Cref{sub:burn}), or when we allow for variable adjustment parameters, target ratio and tolerance distributions as is the case in practice (\Cref{sec:robustness_app}).\par
Technically, our paper is most closely related to the literature on complex dynamics that have been previously studied in various theoretical settings \cite{Gal11,Piliouras14,Piliouras23} and, recently, also in market \cite{Che21} and blockchain settings \cite{Leonardos21,Leonardos2022}. Our current work extends these frameworks to include rigorous guarantees on market-liveness and long-term system performance. In stark contrast to conventional equilibrium analysis, gaining insight on non-equilibrium regimes is integral in the study of systems with strong practical orientation \cite{Leonardos2019}.

\paragraph*{Outline}
In \Cref{sec:model}, we develop the framework to describe dynamic MEV extraction rates, and analyse their evolution and properties under stylised mathematical assumptions in \Cref{sec:analysis}. In \Cref{sec:opt_performance,app:appendix}, we study the performance and robustness of the dynamic MEV sharing mechanism under various market conditions. To conclude the paper, we discuss the economic considerations that arise from implementing dynamic MEV extraction rates in practice and summarise our findings along with directions for future work in \Cref{sec:discussion,sec:conclusion}.

\section{Model}
\label{sec:model}

We consider a two-sided, decentralised market consisting of users (or MEV generators) who submit transactions on the one side and miners (or MEV extractors) who process these transactions on the other side. Here, the term \emph{miners} is used as a general term to describe all the entities involved in the MEV supply chain including searchers, builders and proposers.\footnote{We will also use the the equally general term \emph{block producers} interchangeably.} This does not affect our analysis, since the value extracted from the users' transactions flow through all those involved entities. The market is governed by a protocol which can stipulate the MEV extraction rate, $\la\in [0,1]$. Higher (lower) values of $\la$ indicate more (less) aggressive extraction by miners and, hence, a more (less) hostile environment for users. Equivalently, one may think of $1-\la$ as the fraction of MEV that is returned to users through rebates, direct payments or other on or off-chain vehicles.\par
Each user and miner is characterised by their tolerance on $\la$. Thus, if a user (miner) sees an MEV extraction rate higher (lower) than what they can tolerate, they do not participate in the market. The \emph{tolerance distributions} of users and miners are denoted by $F(\la)$ and $G(\la)$, respectively. Thus, in expectation, the proportion of users that participate in the market for a given $\la$ is $\bar{F}(\la)$, where $\bar{F}(\la):=1-F(\la)$ and the proportion of miners is $G(\la)$. We assume that the support of $F,G$ is included in $[0,1]$ and that both $F,G$ are strictly increasing and differentiable with probability density functions $f=F',g=G'$, respectively.\footnote{The strict monotonicity and differentiability assumptions are common in the economics literature and only serve to avoid edge cases that complicate the analysis without adding much intuition. Our results readily extend to more complex distributions. The assumption that the support of $F,G$ is included in $[0,1]$ implies that all users (miners) drop out of the market if $\la=1$ ($\la=0$). In our model, this is just an abstraction. Real models may allow for different distributions with a positive probability mass at $\la=0$ or $\la=1$. In the context of competing decentralised exchanges or in the PBS framework, this abstraction matches reality since users or miners can choose another service (market maker) that offers a more favourable extraction rate.}

\par
To incorporate the MEV extraction rate in the protocol design space, we make the following design choices:
\begin{itemize}[leftmargin=*]
\item \emph{Dynamic updates:} the MEV extraction rate is regularly updated to reflect prevailing market conditions. Updates may take place in regular intervals, e.g., after every block or epoch. This generates a sequence of MEV extraction rates, $(\la_t)_{t\ge0}$, where $t$ indicates the time (e.g., block height or epoch).
\item \emph{Intensity of updates:} The intensity of the updates between consecutive blocks is regulated by a parameter, $\eta$, that can be chosen by the designer.
\item \emph{Target function:} the goal of the mechanism is to achieve a desired (optimal) distribution of MEV between users and miners that will optimally balance their participation at a desired ratio, denoted by $w$ (target). This means that the designer is seeking to find a $\la^*$ such that $\bar{F}(\la^*)=w\cdot G(\la^*)$. Equivalently, if $m_t$ is the MEV at period $t$, then the above equation suggests that the MEV shares distributed to users and miners are equal to a predetermined balance ratio.\footnote{The MEV at period $t$ may refer to the MEV from a transaction bundle processed by a  rollup or relay operator, the MEV of a block proposed through the MEV-Boost auction or even the MEV of a single transaction when considering direct rebates to users, e.g., \cite{Flashbots_share,wu2024strategic,Flashbots_suave}.}
\item \emph{Market liveness}: the boundary MEV extraction rates of $\la=0$ and $\la=1$ correspond to steady states in which the two-sided market collapses to having either only miners ($\la=1$) or only users ($\la=0$). In this case, the market cannot remain \emph{live}. As mentioned above, the exact values of $\la$ at which this happens is an abstraction (that, nevertheless, matches reality in many practical cases under consideration) but which can be customised to different boundary values without affecting the analysis.
\end{itemize}
Based on the above desiderata, and inspired by EIP-1559 transaction fee design, we consider the following dynamic update rule to govern the evolution of the MEV extraction rates, $\la_t$:
\[\la_{t+1}:=\la_t+\eta\cdot m_t\cdot \la_t \cdot (1-\la_t)\cdot \left(\bar{F}(\la_t)- w\cdot G(\la_t)\right).\]
In the update function of the right hand side, the terms $\la_t, (1-\la_t)$ ensure that $\la=0$ and $\la=1$ are fixed points of the system and the term $\bar{F}(\la_t)- w\cdot G(\la_t)$ serves as an indicator of the current performance of the system, i.e., how far the system is from its objective (recall that $m_t$ is the MEV shared). We can rewrite the above formula as 
 \begin{equation}\label{eq:mev_dynamics}
\lambda_{t+1}=\lambda_t + \eta_t \lambda_t(1-\lambda_t)\Delta_w(\la_t) \tag{$\lambda$-MEV}    
\end{equation}
where $\Delta_w(\la_t) := \bar{F}(\lambda_t)-w G(\lambda_t)$ and $\eta_t = \eta \cdot m_t$. Whenever obvious, we will omit the dependence of $\Delta_w(\la)$ on $w$, and we will simply write $\Delta(\la)$ or $\Delta(\la_t)$ when time is relevant. 
For our purposes, it will also be helpful to define the update function $h: [0,1] \to [0,1]$ as
\begin{equation}
\label{eq:update}
h(\lambda) = \lambda+\eta\lambda(1-\lambda)\Delta_w(\la).   
\end{equation}

\section{Evolution of the Dynamic MEV Extraction Rates}
\label{sec:analysis}

The analysis of the dynamical mechanism, defined by \eqref{eq:mev_dynamics}, will be our main focus. As mentioned above, the mechanism in \eqref{eq:mev_dynamics} is characterised by the choice of two design parameters: the adjustment quotient, $\eta$, which regulates the intensity of the updates, and the target, $w$, which is the desired participation ratio between users and miners in the system (equivalently the desired MEV sharing between them). \par 
In this part, we analyse the system mathematically under the assumption of stable market conditions. This entails constant MEV, i.e., $m_t=m$ for some $m>0$, which in turn implies a constant adjustment paremeter, $\eta_t=\eta$ (propertly rescaled) and stable tolerance distributions $F,G$.\footnote{We consider the case of time varying $m_t,F$ and $G$ in our simulations (\Cref{app:appendix}) and demonstrate that the behaviour of the system remains qualitatively equivalent to the one presented here. Note also, that our mathematical analysis considers the expected behaviour of the system, i.e., known values for $F,G$ rather than sample estimates as in reality. However, this case is also covered by our simulations and shown to produce equivalent results.} 
Our first observation is that under our working assumptions, the target function $\Delta(\la)$ satisfies some useful properties.
\begin{property}[Properties of target function $\Delta(\la)$]\label{pro:property}
   Let $w>0$, and assume that $F,G$ are differentiable and strictly increasing with support in $[0,1]$. Then, the function $\Delta(\la):=\bar{F}(\la)-wG(\la)$ satisfies the following properties:
\begin{enumerate}[label=(\roman*),leftmargin=*]
       \item $\Delta(\la)$ is strictly decreasing for $\la\in [0,1]$ with $\Delta(0)=1$, and $\Delta(1)=-w$.
       \item There exists a unique $\la^*\in (0,1)$ such that $\Delta(\la^*)=0$.
       \item There exists a function $K:[0,1]\to \mathbb{R}$ so that 
       $\Delta(\la)=(\la-\la^*)\cdot K(\la)$, and $\lim\limits_{x\to \la^*}K(x)<+\infty$.
   \end{enumerate}
\end{property}
\begin{proof}
Properties $(i), (ii)$ and the existence part of $(iii)$ are straightforward and their proof is omitted. To see why the second part of $(iii)$ holds, observe that since $F,G$ are differentiable in $(0,1)$, de l'H\^{o}pital's rule implies that 
\begin{align*}
  \lim_{\la\to \la^*}K(\la) &=\lim_{\la\to \la^*}\frac{\Delta(\la)}{\la-\la^*}=  \lim_{\la\to \la^*}\left(-f(\la)-wg(\la)\right) = -f(\la^*)-wg(\la^*)<+\infty,
\end{align*}
which holds since the probability distribution functions exist and are continuous by assumption.
\end{proof}

Based on the above, our first observation is that, by construction, \eqref{eq:mev_dynamics} has 3 fixed points. These are $\la =0$, $\la=1$ and a \emph{unique} interior fixed point $\la^*\in(0,1)$ which corresponds to optimal state of the system, $\Delta_w(\la^*)=0$. The first two fixed points are considered \emph{degenerate states} of the system: the former corresponds to no MEV extraction and, thus, miners are not willing to participate. The latter corresponds to a state in which MEV extraction rate is so high that no user is willing to interact with the blockchain.\par
We are primarily interested in analyzing the evolution of the dynamics with respect to the third fixed point, i.e., the unique interior fixed point which corresponds to the target levels of user/miner participation (optimal state). Before we proceed with our results about the behavior of \eqref{eq:mev_dynamics}, we introduce a necessary definition.

\begin{definition}[Directionally stable fixed point]
\label{def:stable}
Let $\{\la_t\}_{t \ge 0}$ be a one-dimensional dynamical system determined by a function $h: \mathbb{R} \to \mathbb{R}$, so that $\la_{t+1}:=h(\la_t)$ for every $t\ge0$, and let $\la^*$ be a fixed point of $h$, i.e., $h(\la^*)=\la^*$. Then $\la^*$ is called \emph{directionally stable} for $\{h_t\}_{t \ge 0}$, if for every $t \ge 0$ such that $\la_t \neq \la^*$ it holds that $\left(h(\la_t)-\la_t \right)/(\la_t - \la^*)<0$.
\end{definition}

In other words, if a fixed point is \emph{directionally stable} then the dynamical system moves towards the direction of this fixed point at every iteration. We will now proceed to show that in \eqref{eq:mev_dynamics}, the fixed point $0< \la^* <1$ is \emph{directionally stable}.

\begin{lemma}
\label{lem:stable}
The unique interior fixed point, $\la^*$ of the 
\eqref{eq:mev_dynamics} is directionally stable, i.e., it holds that $h(\la)>\la$ whenever $0<\la<\la^*$ and $h(\la)<\la$ whenever $\la^*<\la<1$.
\end{lemma}

\begin{proof}
First, for $\la>\la^*$ it holds that $\Delta_w(\la)<0$ since $\Delta_w(\la^*)=0$ and $\Delta_w$ is a decreasing function and hence $h(\la)-\la = \la(1-\la)\Delta_w(\la)<0$. Thus, $(h(\la)-\la)/(\la - \la^*)<0$. Similarly, for $\la< \la^*$ it holds that $\Delta_w(\la)>0$ and subsequently $(h(\la)-\la)/(\la - \la^*)<0$.
This concludes the proof that $\la^*$ is \emph{directionally stable}.   
\end{proof}

Directional stability secures that whenever a trajectory reaches a point close to a degenerate state, i.e.,  $\la=0$ or $\la=1$, the next step will move towards the fixed point $\la^*$. Due to the factors $\la$ (resp. $(1-\la)$), the updates close to the boundaries will be small. This trend will continue until the trajectory reaches or overshoots the fixed point $\la^*$. \par
An important implication of the previous observation is that the system can reach the degenerate state $\la=0$ at time $t+1$ only when $\la_t>\la^*$, and similarly, it can reach the degenerate state $\la=1$ only when $\la_t<\la^*$. Further elaborating on this, we can derive a condition on the intensity of system updates, to ensure that these two boundary fixed points will \emph{never} be reached by the dynamics, for any interior starting point, $\la_0\in(0,1)$. 

\begin{theorem}(Market-Liveness)
\label{thm:market_liveness}
For any pair of tolerance distributions, $F,G$, target $w$, and starting point $\la_0\in (0,1)$, it holds that $\la_t\in(0,1)$ for any $t>0$ as long as    
$$\eta<\min{\left\{\frac{1}{w(1-\lambda^*)},\frac{1}{\la^*}\right\}}.$$
where $\la^*$ is the unique solution of the equation $\Delta(\la)=0$ in $(0,1)$.
\end{theorem}

\begin{proof}
We first consider the case of avoiding the fixed point at $\la=0$. If $0<\la_t<\la^*$ for some $t\ge0$, then by \Cref{lem:stable}, it holds that $\la_{t+1}=h(\la_t)>\la_t$ which implies, in particular, that $h(\la_t)\neq0$. Thus, for the dynamics to reach $0$, it must be the case that $\la_t>\la^*$ and $h(\la_t)\le 0$ (recall that if $h(\la_t)<0$, then the $\la_{t+1}$ is projected back to $0$). To exclude this case, we need to ensure that $h(\la_t)>0$ for all $\la>\la^*$ which can be written as
\begin{equation}\label{eq:star}
    \Delta(\la)>\frac{-1}{\eta(1-\la)}, \quad \text{for all } 1>\la> \la^*.
\end{equation}
Since $\Delta(\la)=\bar{F}(\la)-wG(\la)>-w$, and $\frac{-1}{\eta(1-\la^*)}>\frac{-1}{\eta(1-\la)}$ for all $1>\la>\la^*$, a sufficient condition for the inequality in \eqref{eq:star} to hold is that $-w >\frac{-1}{\eta(1-\la^*)}$ or equivalently that $\eta<\frac{1}{w(1-\la^*)}$. The case of the fixed point at $\la=1$ is similar. In particular, if $\la^*<\la_t<1$ for some $t\ge0$, then, by \Cref{lem:stable}, the dynamics will decrease in the next step. Thus, the only case for \eqref{eq:mev_dynamics} to reach $1$ is that $\la_t<\la^*$ and $h(\la_t)\ge1$. As before, to exclude this case, we need to ensure that $h(\la_t)<1$ for all $0<\la<\la^*$ which can be written as

\begin{equation}\label{eq:star2}
    \Delta(\la)<\frac{1}{\eta\la}, \quad \text{for all } 0<\la <\la^*.
\end{equation}
Since $\Delta(\la)=\bar{F}(\la)-wG(\la)<1$, and $\frac{1}{\eta\la}>\frac{1}{\eta\la^*}$ for all $0<\la<\la^*$, a sufficient condition for the inequality in \eqref{eq:star2} to hold is that $1<\frac{1}{\eta\la^*}$ or equivalently that $\eta<\frac{1}{\la^*}$. Putting these two conditions together, we obtain the claim, namely that $\eta<\min{\left\{\frac{1}{\la^*},\frac{1}{w(1-\la^*)}\right\}}$.    
\end{proof}

\Cref{thm:market_liveness} provides a sufficient condition for the values of the control parameter $\eta$ for which \eqref{eq:mev_dynamics} never reaches the degenerate states $\la\in\{0,1\}$. We next turn our attention to the evolution of the system for lower values of $\eta$, i.e., for values of $\eta$ below the upper bound of \Cref{thm:market_liveness}. As we show in \Cref{thm:convergence}, at the other extreme, i.e., for sufficiently low values of $\eta$, the \eqref{eq:mev_dynamics} converge to the fixed point $\la^*$.

\begin{theorem}[Convergence to $\la^*$]
\label{thm:convergence}
For any pair of tolerance distributions, $F,G$, target $w$ and starting point $\la_0\in(0,1)$, the sequence $\{\la_t\}_{t>0}$ generated by \eqref{eq:mev_dynamics} converges to $\la^*$, for any control parameter $\eta>0$ such that
\[
\eta\le\inf_{\lambda\neq\lambda^*}\frac{\left(\lambda^*+\lambda\right)}{\lambda^2(1-\lambda)K(\la)}.
\]
\end{theorem}

\begin{proof}
We will prove that the function $\Phi(\lambda):=(\ln {\lambda} - \ln {\lambda^\ast})^2$ decreases along any trajectory of the dynamics, i.e., that $\Phi(h(\lambda))<\Phi(\lambda)$ for any $\lambda \in (0,1)$ with $\lambda \neq \lambda^*$. This is equivalent to showing that 
\begin{equation}\label{eq:target}
(\ln h(\lambda) - \ln \lambda^\ast)^2 < (\ln \lambda - \ln \lambda^\ast)^2, \text{ for any }\lambda \neq \lambda^\ast.
\end{equation} In other words, we will show that $\Phi(\lambda)$ acts as a potential function for the \eqref{eq:mev_dynamics} which directly implies convergence. To proceed, we rewrite the above inequality as
$(\ln h(\la) - \ln \la^\ast)^2 - (\ln \la - \ln \la^\ast)^2 = (\ln h(\la) - \ln \la) \cdot (\ln h(\la) + \ln \la  - 2\ln \la^\ast)
= \ln\left(\frac{h(\la)}{\la} \right)\cdot \ln \left( \frac{\la h(\la)}{(\la^\ast)^2} \right).
$
We will consider two cases, depending on whether $\la<\la^\ast$ or $\la>\la^\ast$:
\begin{itemize}[leftmargin=*]
\item $\la<\la^*$: In this case, it holds that $h(\la)>\la$, see \Cref{def:stable} and \Cref{lem:stable}, which implies that $h(\la)/\la>1$. Hence, $\ln\left(\frac{h(\la)}{\la} \right)>0$. Thus, to obtain the desired inequality, we need to select $\eta>0$ so that the argument of the second $\ln$ on the right hand side of the above equation is less than $1$, i.e., $\la h(\la)/(\la^*)^2<1$. Solving this inequality for $\eta$ yields the inequality $\eta< \frac{(\la^*-\la)(\la^*+\la)}{\la^2(1-\la)\Delta(\la)}$, where we used that $\Delta(\la)>0$ for all $\la<\la^*$. Since this inequality must hold for any $\la<\la^*$, we obtain the threshold
$\eta \le \inf\limits_{\la<\la^*}\frac{(\la^*-\la)(\la^*+\la)}{\la^2(1-\la)\Delta(\la)}$.
It remains to show that the threshold on the right hand side remains strictly positive for any distributions. Using \Cref{pro:property}, we can rewrite the threshold as $\inf\limits_{\la<\la^*}\frac{\la^*+\la}{\la^2(1-\la)K(\la)}$. Since $\Delta(0)=1$, it follows that $K(0)=1/\la^*$ which implies that $\lim\limits_{\la\to 0^+}\frac{\la^*+\la}{\la^2(1-\la)K(\la)}= +\infty$, and
$\lim\limits_{\la\to (\la^*)^-}\frac{\la^*+\la}{\la^2(1-\la)K(\la)}= \frac{2}{\la^*(1-\la^*)\cdot\lim\limits_{\la \to (\la^*)^-}K(\la)}$. 
Thus, since $\lim\limits_{\la \to (\la^*)^-}K(\la)$ exists and is less than $+\infty$ by \Cref{pro:property}, the claim follows.

\item $\la>\la^*$: in this case, it holds that $h(\la)<\la$, see \Cref{def:stable} and \Cref{lem:stable}, which implies that $h(\la)/\la<1$. Hence, $\ln\left(\frac{h(\la)}{\la} \right)<0$. Thus, to obtain the desired inequality in \eqref{eq:target}, we need to select $\eta>0$ so that the argument of the second $\ln$ on the right hand side of the above equation is larger than $1$, i.e., $\la h(\la)/(\la^*)^2>1$. Solving this inequality for $\eta$ yields the same inequality as above (note that this time, $\Delta(\la)<0$).
\end{itemize}
Putting the two cases together, we have that the \eqref{eq:mev_dynamics} converge to $\la^*$ whenever $0<\eta\le \inf\limits_{\la\neq\la^*}\frac{(\la^*+\la)}{\la^2(1-\la)K(\la)}$ as claimed.
\end{proof}

Together with \Cref{thm:market_liveness}, \Cref{thm:convergence} paints the picture for extreme values of update intensities, i.e., either sufficiently high (degenerate market), in which the system collapses by reaching either of the two boundary fixed points, or sufficiently low (lazy updates), in which case the system stabilizes to the unique optimal state. Intuitively, the sufficiency condition of \Cref{thm:convergence} suggests that convergence to $\la^*$ occurs when the intensity of the updates is small enough \emph{in comparison} to current market conditions as these are captured by the function $\Delta(\la)$, i.e., the distributions of users' and miners' tolerances. However, low values of the control parameter $\eta$ correspond to a system with slow responses and come at the trade-off of slow responses to otherwise fast-changing market conditions. Thus, such values may be irrelevant for practical applications \cite{Reijsbergen21}.\par
This leaves open the case in which this condition is not satisfied. Interestingly, as we show next, in this case, the system may exhibit arbitrarily complex behavior. Specifically, in \Cref{thm:chaotic}, we prove that for any value of $\eta$, there exist distributions $F,G$, satisfying our assumptions, for which the system does not converge to $\la^*$. In particular, based on the aggressiveness of the updates, the system can become provably chaotic in the sense of \emph{Li-Yorke chaos} \cite{Li75}. Counterintuitively, as we show in \Cref{sec:opt_performance}, this chaotic behavior is not necessarily detrimental to system performance. Although chaotic updates in the MEV extraction rates make their evolution unpredictable, they keep the system non-degenerate, in the sense that they remain bounded away from $0$ and $1$, thus, ensuring participation of both users and miners. Before we proceed to prove that the \eqref{eq:mev_dynamics} become Li-Yorke chaotic, we introduce new necessary notations and definitions. Our technical tools in this part extend the framework of \cite{Leonardos21}.

\begin{definition}[Periodic Orbit and Points]
\label{def:period_points}
Consider a dynamical system, $(\la_t)_{t\ge0}$, defined by a function $h:\mathbb R\to \mathbb R$ so that $\la_{t+1}=h(\la_t)$ for any $t\ge0$. A sequence $\la_1, \la_2, \dots , \la_k$ is called a \emph{periodic orbit of length $k$} of the dynamical system, if $\la_{t+1}= h(\la_t)$ for $1 \le t \le k-1$ and $h(\la_k) = \la_1$. Each point $\la_1, \la_2, \dots , \la_k$ is called periodic point of least period $k$.    
\end{definition}

\begin{definition}[Li-York chaos \cite{Li75}]
\label{def:chaos}
Let $X=[L,U]$ be a compact interval in $\mathbb{R}$ and let $h: X \to X$ define a discrete-time dynamical system $(x_t)_{t \ge 
 0}$ on $X$, so that $x_t \defeq h^t(x_0)$ for any $x_0 \in X$. The dynamical system $(x_{t})_{t \in N}$ is called \emph{Li-Yorke chaotic} if it holds that
\begin{enumerate}
    \item For every $k=1,2, \dots$ there is a periodic point in $X$ with period $k$.
    \item There is an uncountable set $S \subseteq X$ (containing no periodic points), which satisfies the following conditions:
        \begin{itemize}[leftmargin=*]
        \item For every $x \neq y \in S$, $\lim\limits_{t \to \infty} \sup \abs{h^t(x)-h^t(y)}{} >0$ and $\lim\limits_{t \to \infty} \inf \abs{h^t(x) -h^t(y)}{}=0$.
        \item For every points $x \in S$ and $y \in X$, 
 $\lim\limits_{t \to \infty} \sup \abs{h^t(x)-h^t(y)}{} >0$.
    \end{itemize}
    \end{enumerate}

\end{definition}

Famously, a sufficient condition for a system to be \emph{Li-Yorke chaotic} is to have a periodic point of period of $3$ \cite{Li75}. This implication relies on the Sharkovsky Ordering and and it is actually a special case of Sharkovsky's theorem which is presented next \cite{Sharkovsky64}. 

\begin{definition}[Sharkovsky Ordering \cite{Burns11}]
\label{def:sharkovsky}
The following ordering of the set $\mathbb{N}$ of positive integers, is known as the Sharkovsky ordering: 
$3>5>7>9>\cdots$, $3\cdot2>5\cdot2>7\cdot2>9\cdot2>\cdots$, $3\cdot2^n>5\cdot2^n>7\cdot2^n>9\cdot2^n>\cdots$, and $\cdots>2^n>\cdots>2^2>2>1$.
This is a total ordering, i.e., every positive integer appears exactly once in the list. 
\end{definition}

\begin{theorem}[Sharkovsky \cite{Sharkovsky64}]
\label{thm:sharkovsky}
Let $X \subset \mathbb{R}$ be a compact interval and let $h:X \to X$ be a continuous function. If $h$ has a periodic point of least period $m$, then $h$ has a periodic point of least period $l$ for every $l$ such that $m>l$ in the Sharkovsky ordering.
\end{theorem}

Using these results, we will now formally prove that for every intensity, $\eta$, in the MEV extraction rate updates, there exist tolerance distributions $F,G$ such that \eqref{eq:mev_dynamics} become \emph{Li-York chaotic}.

\begin{theorem}[Chaotic Updates]
\label{thm:chaotic}
For every intensity level, $\eta>0$, there exist continuous distributions $F,G$ such that the \eqref{eq:mev_dynamics} become Li-Yorke chaotic.
\end{theorem}
\begin{figure*}[t]
    \centering
    \includegraphics[width=0.95\linewidth]{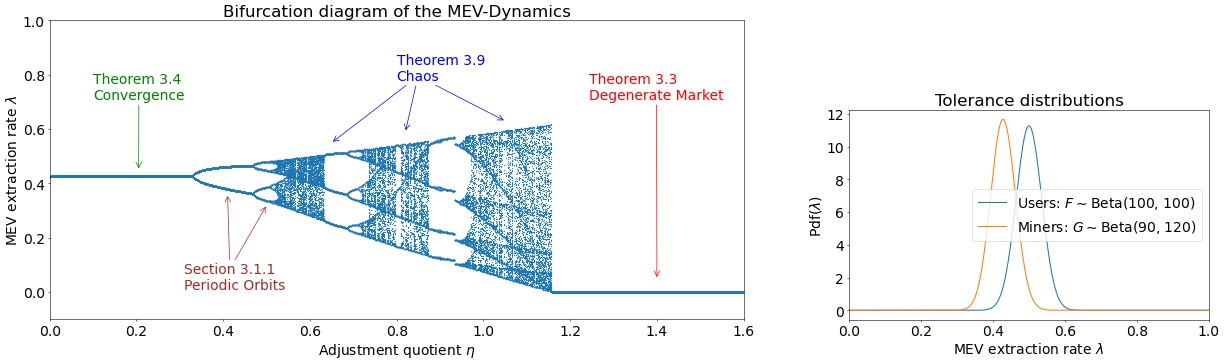}
    \caption{Bifurcation diagram for the \eqref{eq:mev_dynamics} with respect to the adjustment intensity $\eta$ (left panel). The tolerance distributions of miners and users are shown in the panel on the right. For low values of $\eta$, the dynamics converge to $\lambda^*$ (\Cref{thm:convergence}), whereas for large values of $\lambda$, the dynamics reach the boundary and get trapped at the corresponding fixed point (in this case $0$) in which case, the two-sided market collapses (\Cref{thm:market_liveness}). For intermediate values of $\eta$, the dynamics are either provably chaotic (\Cref{thm:chaotic}) or periodic with periods of difference density (\Cref{sub:periodic}).}
    \label{fig:main}
\end{figure*}

\begin{proof}
The proof is constructive and proceeds by creating an instance for which the dynamics have a periodic point of least period $3$. The instance depends on the $\eta$ and thus, it can apply to any $\eta>0$. Specifically, we consider a case in which both users and miners have the same tolerance distributions on $[a,b]$ with $0<a<b<1$. In this case, the update function \eqref{eq:update} can be written as
\[
h(\lambda)= \begin{cases}
\lambda(-\eta\lambda +\eta +1), & 0 \le \lambda < a, \\
\lambda(-\eta\lambda \Delta(\lambda) + \eta\Delta(\lambda) +1), & a \le \lambda <b, \\
\lambda(\eta w \la -\eta w +1) & b \le \lambda \le 1,
\end{cases},
\]
with $\Delta(\la) := \frac{b+wa-\la(1+w)}{b-a}$. We will prove that for any fixed $\eta>0$, there exist $a,b$ such that the dynamical system has a periodic orbits of period 3. Since the update function, $h$ is continuous, we will prove that by finding points $\la_0, \la_1,\la_2, \la_3$ such that \[\la_1= h(\la_0), \,\,\la_2=h(\la_1)=h^2(\la_0),\,\, \la_3=h(\la_2)=h^3(\la_0)\] which satisfy $\la_3 \le \la_0 <\la_1 < \la_2$. Firstly, for $0 \le \la < a$, the update function $h$ is increasing and $\la<h(\la)$. Thus, there exists $0< \la_0<a$ such that $h(\la_0)=a$. Let $\la_1=a, \la_2=h(a)$ and $\la_3=h^2(a)$. Then, 
$h(a)=a+ \eta a(1-a)$ and, hence, $a<h(a)$. Moreover $\Delta(h(a))= 1- \eta c $, where $c=\frac{a(1-a)(1+w)}{b-a}$ and, as $b\rightarrow a^+$, $c \rightarrow \infty$. Lastly, we compute \[h(h(a))= h(a) +\eta h(a)(1-h(a))\Delta(h(a)).\] After some trivial algebra, we get that \[h(h(a))= -(h(a) - (h(a))^2)\eta^2 c + (h(a) - (h(a)^2))\eta +h(a).\]
We want to prove that $\la_3 = h(h(a)) \le \la_0$ which is equivalent to
\[-(h(a) - (h(a))^2)\eta^2 c + (h(a) - (h(a)^2))\eta +h(a) \le \la_0.\] 
This is true as $b \rightarrow a^+$, since $h(a)-(h(a))^2>0$. Thus we get $\la_3 \le \la_0 <\la_1 < \la_2$, which concludes the proof.
\end{proof}
Intuitively, the construction in the proof of \Cref{thm:chaotic} suggests that as the tolerance of users and/or miners becomes more concentrated around some value, then the dynamics are more likely to become chaotic since even small updates in the MEV extraction rate have significant impact in the participation rates of the system.

\subsection{Visualising the dynamic mechanism}\label{sub:visualizations}
In order to gain a better intuition on the results of the previous section, in this part, we provide visualizations of the \eqref{eq:mev_dynamics} that showcase the various cases discussed above.\par 
To instantiate the dynamics, we select specific instances of Beta distributions for both users and miners because these distributions have the desirable properties that (1) their support range is $[0,1]$, matching the values of possible MEV extraction tolerance for miners and users, and (2) they have the expressive capacity to capture many different forms of valuation distributions. The results for a specific instance are provided in the \emph{bifurcation diagram} in \Cref{fig:main}. \par
In the bifurcation diagram, the horizontal axis corresponds to values of the update instensity, $\eta$. It is straightforward to see, that for any fixed value of $\eta$, the \eqref{eq:mev_dynamics} either converge (low values of $\eta)$ or are restricted in a bounded area around $\la^*$. This bounded area expands as $\eta$ increases till it reaches one of the boundaries (in this case $\la=0$) which implies that the dynamics get trapped there (with dire implications for the two-sided market).\par
The instance depicted in the left panel of \Cref{fig:main} captures all cases described in our previous formal analysis. From \Cref{fig:main}, we can also see some regimes, e.g., around $\eta\approx0.4$, in which the dynamics only have \emph{periodic orbits}. In fact, using \Cref{thm:sharkovsky}, we can further argue about the behavior of the system in these cases.

\subsubsection{Regions of periodic orbits}\label{sub:periodic}
To better understand the orbits of \Cref{eq:mev_dynamics}, we visualize the function $h^{(k)}(\la)$ and observe its fixed points, $\la=h^{(k)}(\la)$ for various values of $k\in \mathbb {N}$. In particular, \Cref{thm:sharkovsky} implies that if we can find $k\ge2\in \mathbb N$ such that the function $h(\la)$ has periodic points of least period $k$, 
but \emph{not} of least period $k+1$, then the dynamics generated by $h$ have periodic orbits with all possible periods that precede $k$ in the Sharkovsky ordering. As mentioned above, such regimes are visible around $\eta=0.4$ in \Cref{fig:main}. \par
In \Cref{fig:periodic}, we visualize one such instance that corresponds to areas of periodic behavior for $\eta=0.5$. Note that for $\eta$ slightly larger than $0.6$, the bifurcation diagram shows a range of parameter values for which apparently the only cycle has period $3$. In fact, according to \Cref{thm:sharkovsky}, there must be cycles of all periods there, but they are not stable and therefore not visible on the computer-generated simulation.
\begin{figure}[t]
    \centering
    \includegraphics[width=0.3\linewidth]{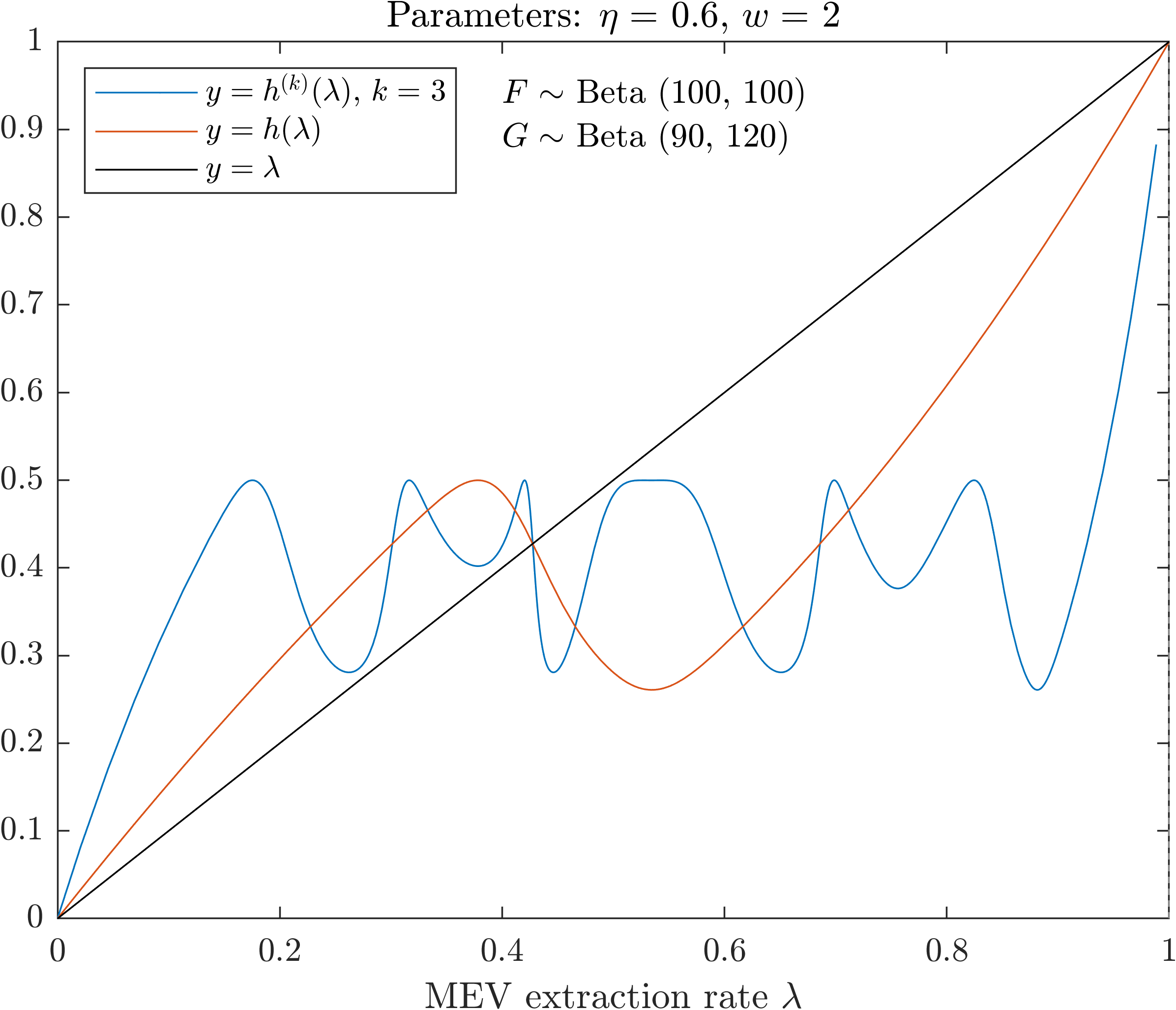}\hspace{10pt}
    \includegraphics[width=0.3\linewidth]{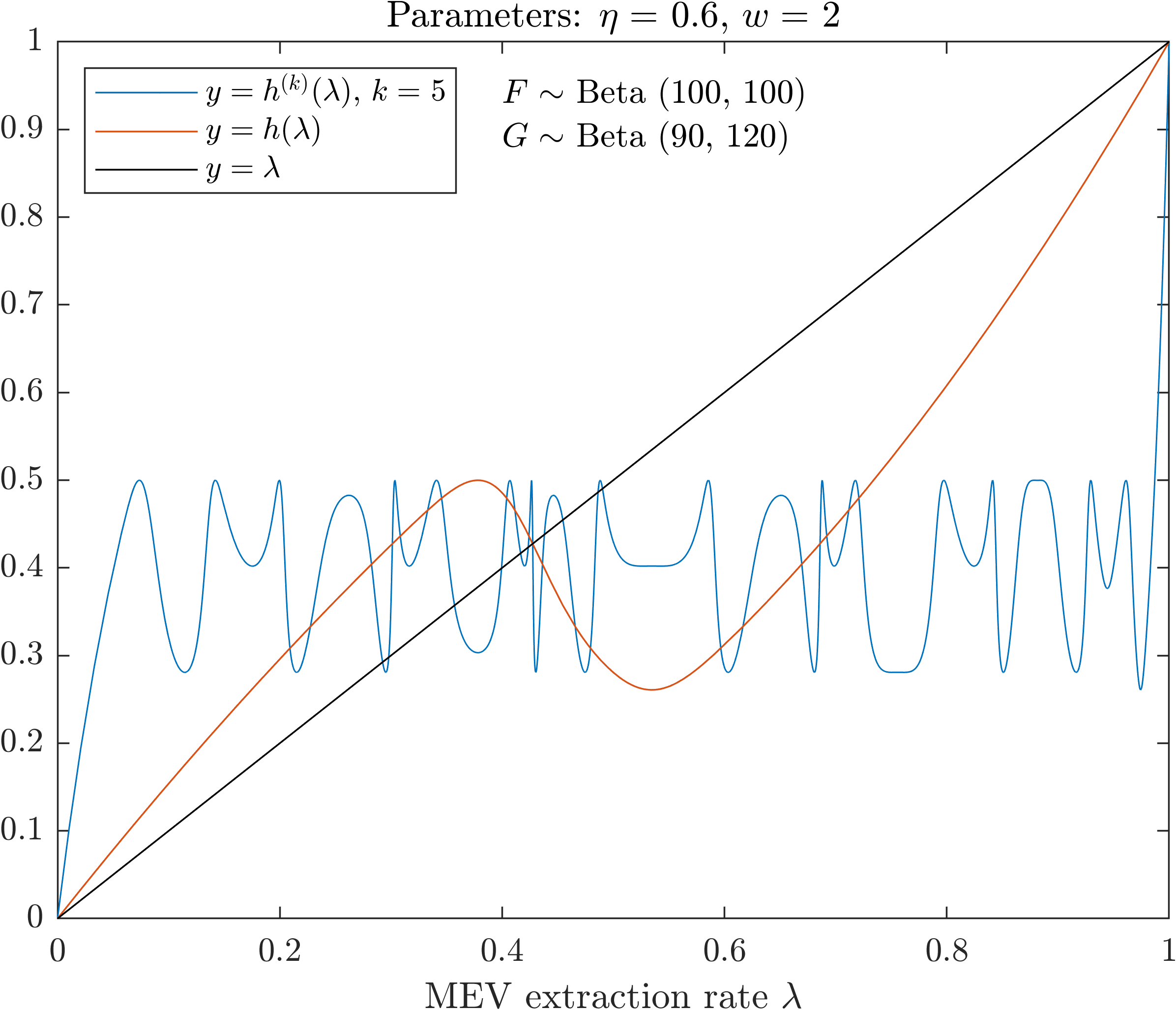}\hspace{10pt}
    \includegraphics[width=0.3\linewidth]{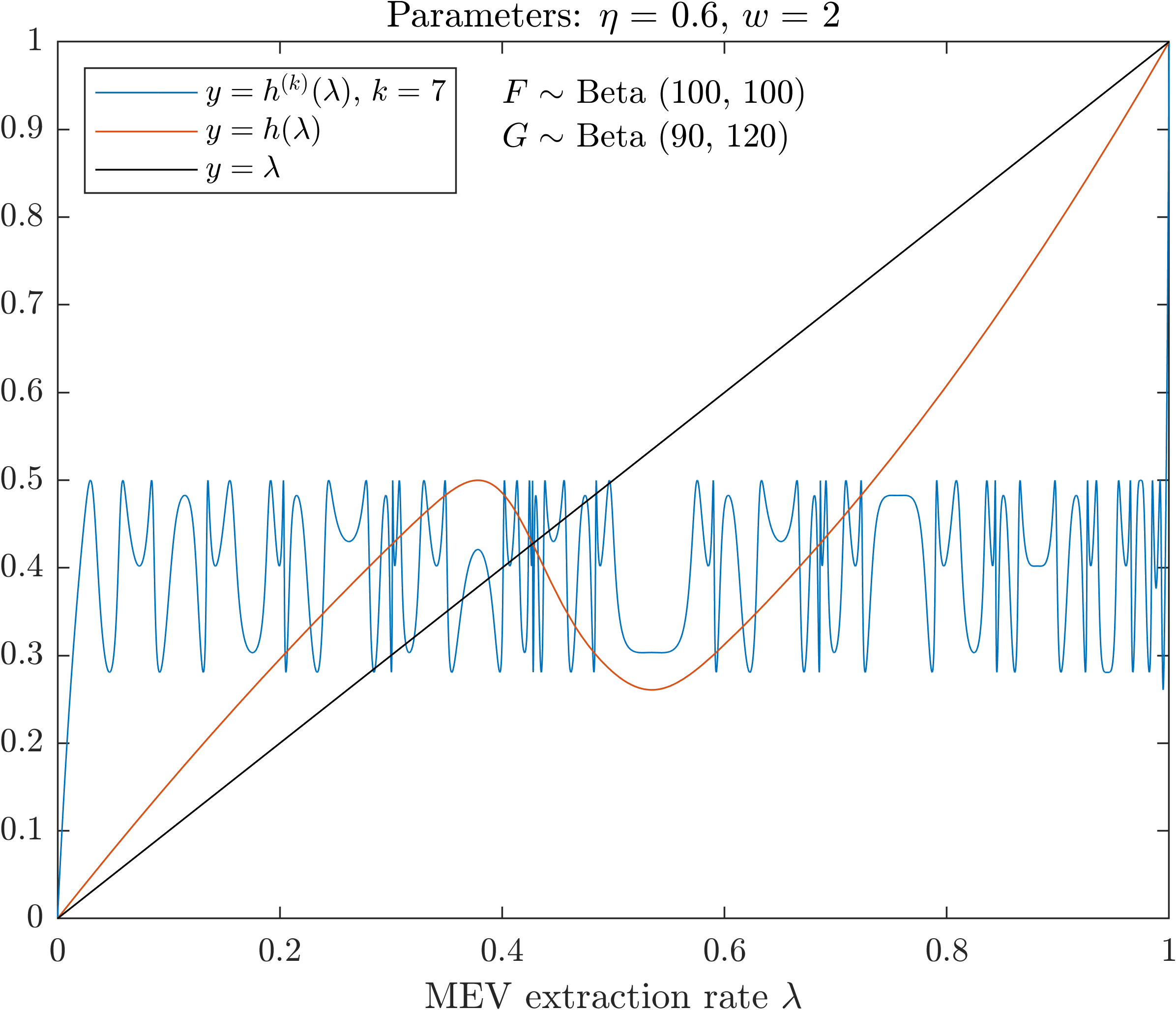}
    \caption{Regions of periodic behavior for the instance of \Cref{fig:main} and $\eta=0.6$. The panels indicate the existence of periodic points of least period $k=5$ and $k=7$, but not of $k=3$ (the period $k$ is denoted in the legends). According to \Cref{thm:sharkovsky}, this means that for these parameters, the \eqref{eq:mev_dynamics} have periodic points of uncountably many periods but are not provably chaotic. Similarly, we can create instances with periods of $k=4,2$ and $1$ but not of $k=8$ etc.} 
    \label{fig:periodic}
\end{figure}


Moreover, in \Cref{fig:bifurcation_range_valuations}, we visualize the long term behavior of \eqref{eq:mev_dynamics} with respect to the range of tolerance distributions. 

\begin{wrapfigure}[14]{r}{0.55\linewidth}
    \centering
    \vspace{-0.94cm}
    \includegraphics[scale=0.3]{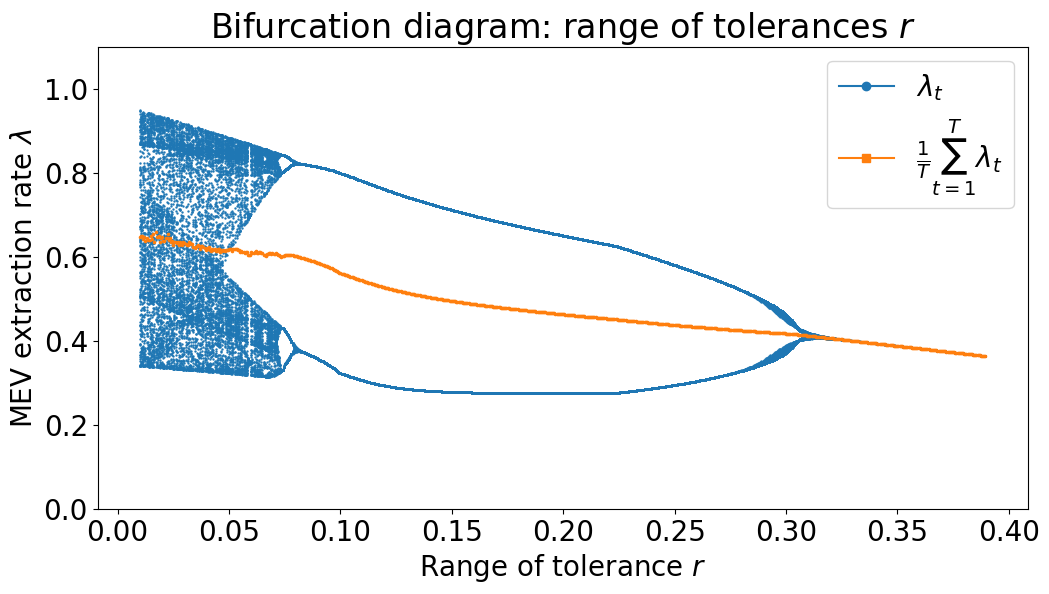}
    \caption{Bifurcation diagram for the \eqref{eq:mev_dynamics} with respect to the range of tolerance distributions. We observe here the route from chaos to order as the range of the tolerance increases.}
    \label{fig:bifurcation_range_valuations}
\end{wrapfigure}
In this case, we keep the adjustment quotient $\eta$ fixed. Nevertheless simulations with different distributions produce substantially the same outcome. Now, the horizontal axis corresponds to the varying parameter $r$, which is the range of the tolerance distributions. For this simulation, we use uniform distributions for the tolerances $F$ and $G$ on the interval $[L_u,U_u]= [0.5-r, 0.5+r]$ and $[L_m,U_m]= [0.4-r, 0.4+r]$, respectively. Intuitively, as the area between the probability of density distributions  $\Bar{F}'$ and $G'$ increases, the dynamics \eqref{eq:mev_dynamics} tend to stabilize. Interestingly, \Cref{fig:bifurcation_range_valuations}, serves also as a visualization of \Cref{thm:chaotic}, which states that for each adjustment quotient $\eta$, there are tolerance distributions $F,G$ such that the \eqref{eq:mev_dynamics} become chaotic. We note that, we use uniform distributions for simplicity, 

\section{Performance}
\label{sec:opt_performance}

In the previous section, we analyzed the long-term performance of \eqref{eq:mev_dynamics}. We studied three properties regarding the evolution of the system; conditions for convergence, conditions for not reaching the boundaries and conditions for chaotic behavior. Having established these, we now turn to the performance of the system in these regimes. As mentioned above, low values of the control parameter, $\eta$, allow the dynamics to stabilize at the fixed point $\la^*$. However, such values may lead to slow responses of the system against fast-changing market conditions and may be rarely applicable in practice. Moreover, the intensity of the updates is only relative to prevailing market conditions which may be subject to continuous changes.\par
Thus, the important question that we need to study is whether the system can exhibit desirable performance even when it is pushed beyond the stable regime, into the densely periodic or chaotic regimes. This is critical to ensure near-optimality, i.e., bounded deviations from target, long-term survival, i.e., market-liveness, and ultimately, persistence of the system under possibly adverse conditions. Interestingly, our analysis will not only establish such desirable behavior in out of equilibrium conditions, but also provide formal guarantess (rigorous proofs) that this will be the case. \par
\subsection{Bounded deviations and average-case performance}
Remember that, based on Sharkovsky's Theorem (\Cref{thm:sharkovsky}) the further left a period point of $m$ is in the Sharkovsky's ordering (\Cref{def:sharkovsky}), the denser the periodic orbits are. Thus, to measure the performance of this system, we study if the target ratio $w$ is achieved on average, or equivalently the time-average behavior of the values of the target function, i.e., $\lim\limits_{T\to\infty}\frac{1}{T}\sum_{t=1}^T\Delta(\la)$.\par
In \Cref{fig:performance}, we visualize the values of $\Delta(\la_t)$ for the instance of \Cref{fig:main}. As the control parameter $\eta$ increases, the average deviations increase. This trend (deviation from average optimal performance) continues until the dynamics hit the boundary leading to non-liveness. The simulations suggest that optimal values of the control parameter $\eta$, as far as achieving the target $w$ on average is concerned, correspond to regimes with chaotic or complex - periodic with dense periods - evolution of the MEV extraction rate. From a practical perspective, such regimes ensure \emph{market-liveness} along with bounded deviations from the target (see \Cref{lem:attracting}).\par
\begin{figure*}[t]
    \centering
    \includegraphics[width=0.95\linewidth]{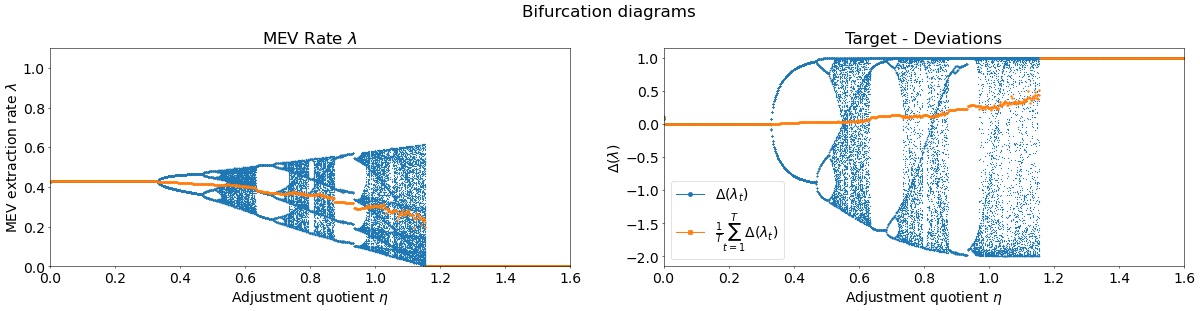}
    \caption{Performance of the \eqref{eq:mev_dynamics} in the instance of \Cref{fig:main}, cf. left panels in both Figures. The simulations use a burn-in period of $200$ iterations followed by $T=200$ iterations that are plotted in the bifurcation diagrams. The right panel shows the values of the target, $\Delta(\la_t)$, for each $t=201,\dots, 400$ in blue dots and the averages over this period in light-colored dots (see also the legend). As $\eta$ grows, the deviations from the target also grow till the dynamics hit the boundary and get absorbed there.}
    \label{fig:performance}
\end{figure*}
In the remainder of this section, we quantify these empirical results regarding system performance. Our first observation is that the sequence of MEV extraction rates, $(\la_t)_{t\ge0}$, remains within a bounded neighborhood of $\la^*$, which depending on the size of $\eta$, can be non-trivial (i.e., strictly between $0$ and $1$) even in the chaotic or densely-periodic cases. This is formally stated in \Cref{lem:attracting}.
 
\begin{lemma}[Attracting range of $\lambda$.]\label{lem:attracting} Consider an instance of the \eqref{eq:mev_dynamics} with $\eta, w>0$ and tolerance distributions $F,G$ with $\lambda^* \in (0,1)$ such that $\Delta(\lambda^*)=0$ as usual. Then, for the sequence of MEV extraction rates $(\lambda_t)_{t\ge0}$ generated by the \eqref{eq:mev_dynamics} it holds that 
\[\max{\left\{0,\lambda^*-\frac{w\eta}{4}\right\}}\le \lambda_t \le \min{\left\{\lambda^*+\frac{\eta}{4},1\right\}}~\text{for all~}t\geq0.\]
\end{lemma}
\begin{proof}
    By \Cref{lem:stable}, we know that $h(\la)>\la$ only if $\la<\la^*$. Thus, $h(\la)= \la+\eta\la(1-\la)\Delta(\la)<\la^*+\eta\cdot \frac14\cdot 1=\la^*+\frac{\eta}{4}$,
where we used that $\Delta(\la)\le 1$ for all $\la\in[0,1]$. Similarly, to determine the lower bound, we need to consider $\la>\la^*$. In this case, we have that 
\begin{align*}
    h(\la)&=\la+\eta\la(1-\la)\Delta(\la)
        >\la^*+\eta\la (1-\la)\cdot(-w)
        >\la^*-\frac{\eta w}{4},
\end{align*}
where we used that $\Delta(\la)\ge-w$ for all $\la\in [0,1]$. 
\end{proof}

Note that the conditions $\la^*-\frac{w\eta}{4}>0$ whenever $\eta<\frac{4\la^*}{w}$ and $\la^*+\frac{\eta}{4}<1$ whenever $\eta<4(1-\la^*)$. By \Cref{lem:attracting}, these inequalities offer an alternative set of market-liveness conditions to those of  \Cref{thm:market_liveness}. \par
The result of \Cref{lem:attracting} is practically relevant when $\eta$ is small enough, so that both comparisons are resolved with the non-trivial bounds, rather than with $0$ and $1$. For such values of $\eta$, the sequence $\la_t$ remains bounded between $\la^*-\frac{w\eta}{4}$ and $\la^*+\frac{\eta}{4}$. This allows us to bound the possible values of $\Delta(\la_t)$, i.e., the possible deviations from the target $\Delta(\la^*)=0$. To do so, we proceed with a concrete expression for function $\Delta(\la)$ that uses the class of beta distributions for the tolerances of both miners and users.\footnote{As mentioned above, beta distributions have the desirable properties that (1) their support range is $[0,1]$, matching the values of possible MEV extraction tolerance for miners and users, and (2) they have the expressive capacity to capture many different forms of valuation distributions.} The result regarding performance is provided in \Cref{thm:performance}.

\begin{theorem}[Bounded Deviations]\label{thm:performance}
    Consider an instance of the \eqref{eq:mev_dynamics} with $F\sim \text{Beta}(a,b), G\sim \text{Beta}(c,d)$ and $\eta,w$ such that $0 < \la^*-\frac{w\eta}{4}<\la^*+\frac{\eta}{4} < 1$. Then, we have
    \[|\Delta(\la_t)|\le \frac{\eta(1+w)}{4}\left(\max_{\la\in (p,q)}{\{f(\la)\}}+w\max_{\la\in (p,q)}{\{g(\la)\}}\right),\]
where $f,g$ denote the probability density functions of $F,G$ and $p=\la^*-\frac{w\eta}{4}, q=\la^*+\frac{\eta}{4}$.
\end{theorem}

The next Lemma offers an exhaustive case analysis for the maximization problems that are involved in \Cref{thm:performance}. Our intention is to use this lemma for values of $p=\la^*-\frac{w\eta}{4}$ and $q =\lambda^* + \frac{\eta}{4}$, to find analytically the maxima of the probability density functions $f$ and $g$ of the beta distributions as used in \Cref{thm:performance}. The proof is provided in the \Cref{appendix}.

\begin{lemma}[Maxima of Beta distributions]\label{lem:beta}
    Let $p,q \in (0,1)$ and $p < q$. Suppose that we are given a function $f : [0,1] \rightarrow {\mathcal{R}}$, such that $f(x) = x^{a-1} \cdot (1-x)^{b-1}$, for some fixed constant values $a, b \in \mathcal{R}$, and for any $x \in [0,1]$. Then we have the following:
    Suppose first that $a+b \not = 2$.
    \begin{itemize}[leftmargin=*]
        \item If $2 > a+b$ and 
         \begin{itemize}[leftmargin=*]
             \item $\frac{1-a}{2-a-b} < p$, then $f$ has the maximum value in $[p,q]$ that is $f(q)$; i.e.~$\forall x \in [p,q] : f(x) \in (0,f(q)]$.
             \item $p \leq \frac{1-a}{2-a-b} \leq q$, then $f$ has the maximum value in $[p,q]$ that is $\tau = \max\{f(p), f(q)\}$; i.e.~$\forall x \in [p,q] : f(x) \in (0,\tau]$.
             \item $\frac{1-a}{2-a-b} > q$, then $f$ has the maximum value in $[p,q]$ that is $f(p)$; i.e.~$\forall x \in [p,q] : f(x) \in (0,f(p)]$.
         \end{itemize}
        \item If $2 < a+b$ and 
         \begin{itemize}[leftmargin=*]
             \item $\frac{1-a}{2-a-b} < p$, then $f$ has the maximum value in $[p,q]$ that is $f(p)$; i.e.~$\forall x \in [p,q] : f(x) \in (0,f(p)]$.
             \item $p \leq \frac{1-a}{2-a-b} \leq q$, then $f$ has the maximum value in $[p,q]$ that is $\gamma = f\left(\frac{1-a}{2-a-b}\right)$; i.e.~$\forall x \in [p,q] : f(x) \in (0,\gamma]$.
             \item $\frac{1-a}{2-a-b} > q$, then $f$ has the maximum value in $[p,q]$ that is $f(q)$; i.e.~$\forall x \in [p,q] : f(x) \in (0,f(q)]$.
         \end{itemize}
    \end{itemize}

    If $a+b = 2$, then under $a>1$, function $f$ reaches it maximum at $x=q$, i.e., $\forall x \in [p,q] : f(x) \in (0,f(q)]$.

    If $a+b = 2$ and $a<1$, then function $f$ reaches it maximum at $x=p$, i.e., $\forall x \in [p,q] : f(x) \in (0,f(p)]$. 
    
    Finally, if $a+b=2$ and $a=1$, then $f(x) = 1$ for all $x \in [p,q]$.
\end{lemma}

\section{Robustness of Findings}
\label{app:appendix}

For the previous rigorous mathematical analysis, and in particular, for the derivation of Li-Yorke chaos, we considered the deterministic evolution of the system which resolves randomness through its expectation. Our extensive simulations, presented next, showcase that the behaviour of the system continues to match the behaviour proved above, even when we lift such assumptions (which are only necessary for the mathematical proofs), i.e., when we perturb the update rule or allow for variable adjustment parameters, target ratio and tolerance distributions.

\subsection{Variations in the update rules}
\label{subsec:expressiveness}
We start by demonstrating that the behavior of the dynamics \eqref{eq:mev_dynamics} is expressive enough, to capture the evolution of a family of update rules. Firstly, we recall that \eqref{eq:mev_dynamics} have three fixed points: $\la=0$ and $\la=1$, which correspond to the degenerate states of the market, and the interior fixed point $\la^*$. However, when the dynamics \eqref{eq:mev_dynamics} tends to the boundaries, i.e., $\la \to 0$, or $\la \to 1$, the updates are slow because of the factors $\la$ (or $(1-\la)$, respectively) in the update rule \eqref{eq:update}. However, the designer would want the updates in one or both of the boundaries to be more aggressive in order  to quickly restore the desired balance between users and miners. Next, we show, that the following update rules exhibit similar behavior compared to the dynamic evolution of \eqref{eq:update}.
\begin{figure}[t]
\centering
    \begin{subfigure}{.475\textwidth}
        \includegraphics[width=\textwidth]{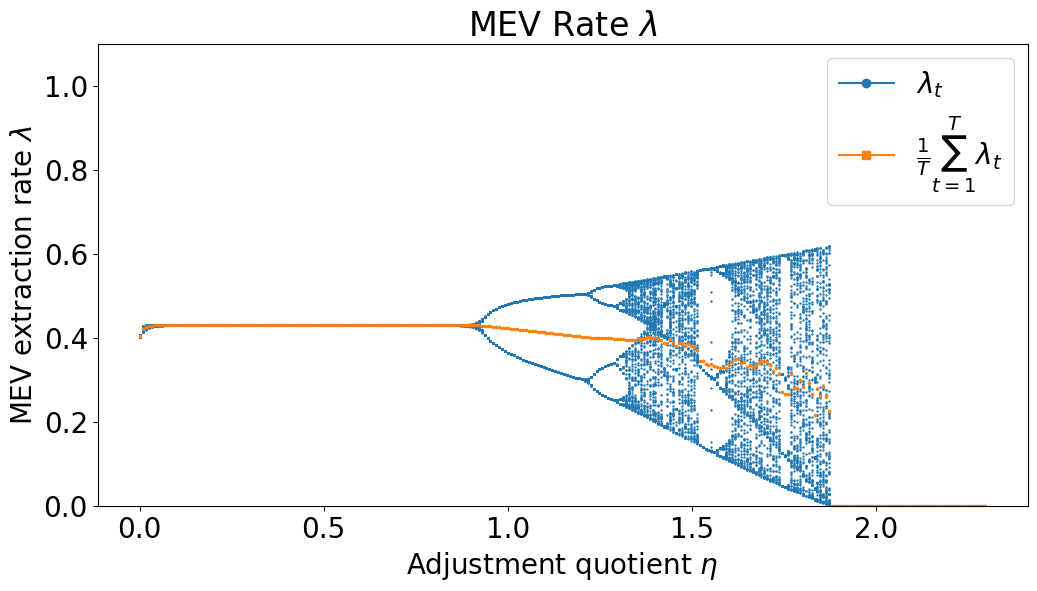}
        \caption{$h(\la) = \la + \eta\la (1-\la) \Delta_w(\la)$}
        \label{subfig:h}
    \end{subfigure}
\hspace{5pt}
    \begin{subfigure}{.475\textwidth}
        \centering
        \includegraphics[width=\textwidth]{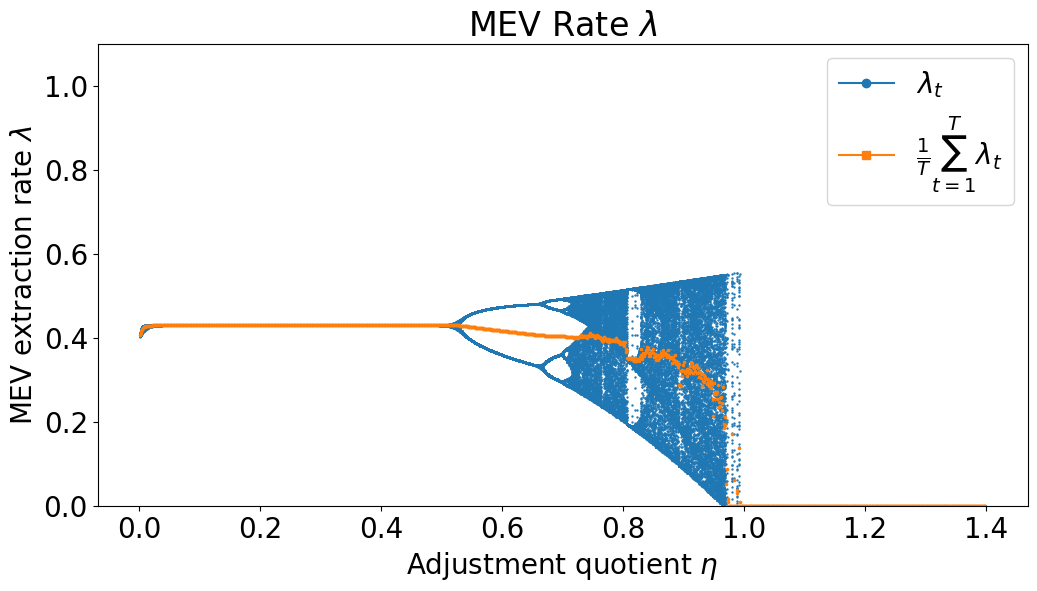}
        \caption{$h_1(\la) = \la + \eta \la \Delta_w(\la)$}
        \label{subfig:h1}
    \end{subfigure}
    \vskip\baselineskip
    \begin{subfigure}{.475\textwidth}
        \centering
        \includegraphics[width=\textwidth]{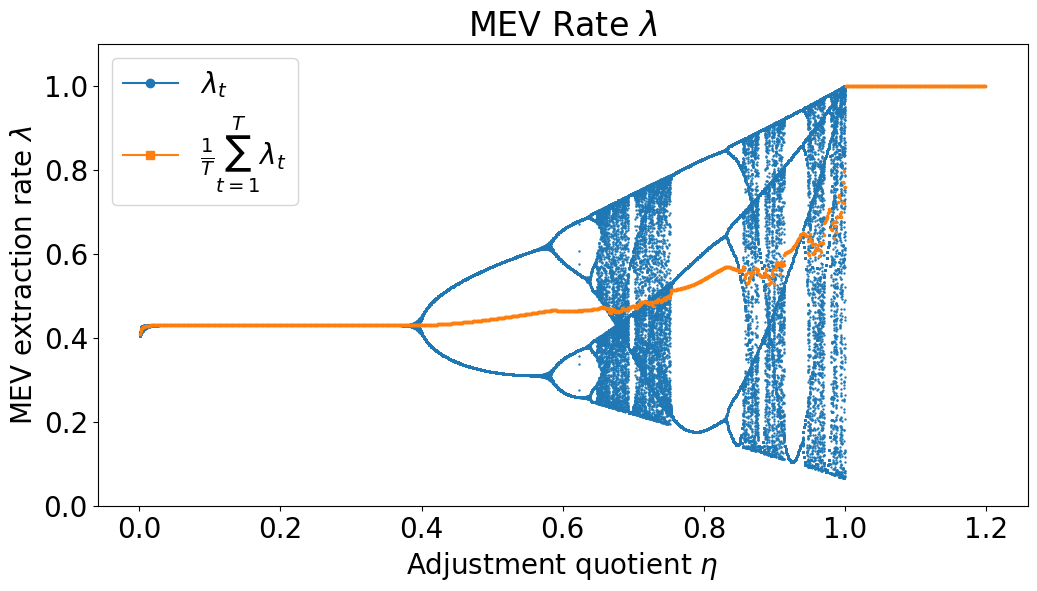}
        \caption{$h_2(\la) = \la + \eta (1-\la) \Delta_w(\la)$}
        \label{subfig:h2}
    \end{subfigure}
\hspace{5pt}
    \begin{subfigure}{.475\textwidth}
        \centering
        \includegraphics[width=\textwidth]{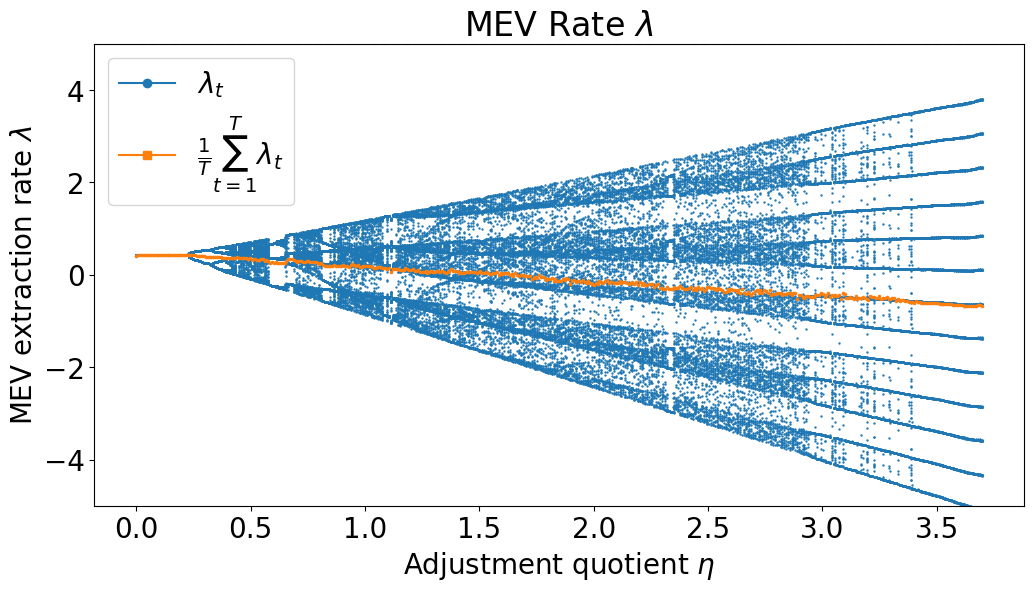}
        \caption{$h_3(\la) = \la + \eta \Delta_w(\la)$}
        \label{subfig:h3}
    \end{subfigure}
    \caption{In all the above simulations, we use Normal distributions for both the tolerance of users and miners on the extracted MEV rate. Specifically $F \sim \mathcal{N}(\mu = 0.4, \sigma^2=0.01)$ and $G \sim \mathcal{N}(\mu = 0.5, \sigma^2=0.01)$. In our toy example, the targeted ratio between users and miners participation is $w=1.6$. 
    }
    \label{fig:different_updates}
\end{figure}

The purpose of the above simulations is to figure out if by removing the fixed points $\la=0$ and $\la=1$ and allowing $\la$ to take arbitrary value, the dynamics maintain a similar behavior. In order to do so, for all the update rules $h, h_1, h_2, h_3$, the interior fixed point $\la^* \in (0,1)$ must remain directionally stable. To achieve that, for the update rule $h_1(\la) = \la + \eta \la \Delta_w(\la)$, for each $\la$, it should hold $\la \ge 0$; otherwise, if $\la<0$, we project this value to $0$. In order to check that this is mandatory, suppose that we allow negative value and for some time $t$, $\la_t<0$. This would imply that $h_1(\la_{t+1})<\la_t<0$. For that extraction rate, it holds $\la_t< \la^*$, and thus the dynamics become unstable since  $h_1(\la_t)-\la_t/\la_t - \la^* >0$. Similarly, for the update rule $h_2$, in order to maintain the directional stability of the interior fixed point $\la^*$ the updates need to be projected to $(- \infty,1]$. For the update rule $h_3(\la)$, in order to maintain the directional stability of $\la^*$, we do not have to restrict the values of $\la$, because it actually is directionally stable, see \Cref{subfig:h3}. \par

Moreover, by comparing the dynamic evolution of \eqref{eq:mev_dynamics}  with the dynamic evolution of the update rules $h_1, h_2, h_3$ in \Cref{fig:different_updates}, we observe that, in the latter case, the dynamics collapse faster to the boundaries $\la=0$ or $\la=1$ since the updates are more aggressive. That is because, for $\la \in (0,1)$, it holds that $\la (1-\la)< \la $ and $\la (1-\la)< 1-\la$. Moreover, for each of those update rules, we can prove similar bounds for the respective \textit{Market Liveness}, see \Cref{thm:market_liveness} about the convergence to $\la^*$, and \Cref{thm:convergence} about the chaotic updates, and \Cref{thm:chaotic}. That is, we can also find the values of the parameters for which analogous theorems hold for the systems $h_1, h_2$ and $h_3$. We also remind that the former theorems do not rely on the tolerance distributions $F,G$ and the results hold regardless of the exact distribution.

\subsection{MEV-burn}
\label{sub:burn}
An important variation of the above mechanism corresponds to the case in which some part of the generated MEV, $m_t$ is burned, i.e., not shared between users and block producers. To formalise this, we introduce a \emph{burn} parameter, $k\in[0,1]$, which determines the fraction of MEV allocated to block producers. This changes the \eqref{eq:mev_dynamics} as follows:
\begin{equation}
    \lambda_{t+1}=\lambda_t+\eta m_t\lambda_t(1-\lambda_t)\left(\bar{F}(\lambda_t)-wG(k\lambda_t)\right)\tag{MEV-burn}.
\end{equation}
In this case, users receive $(1-\lambda_t)$ fraction of the generated MEV, whereas block producers only receive only $k\lambda_t$ fraction of the generated MEV, with the rest, $(1-k)\lambda_t$, being burned.\par
Interestingly, the observed behaviour of the dynamic mechanism continues to hold also in this case. In our simulations (not all presented here), we consider several scenarios in which burning can be constant or even variable over time, e.g., to reflect changes in transaction fees. In \Cref{fig:burn}, we present one such simulation in which the burning parameter, $k$, i.e., the fraction that builders actually retain from the allocated MEV, is randomly sampled in $[0.8,0.95]$ at every iteration. As we see, the introduction of a random burn parameter increases the chaotic behaviour of the mechanism (possibly eliminating the periodic regimes) but qualitatively the dynamics remain equivalent. If $k$ is deterministic, then the main Theorems above can be recovered formally as the dynamics remain directionally stable (not shown here).

\begin{figure}[t]
    \centering
    \includegraphics[width=0.96\linewidth]{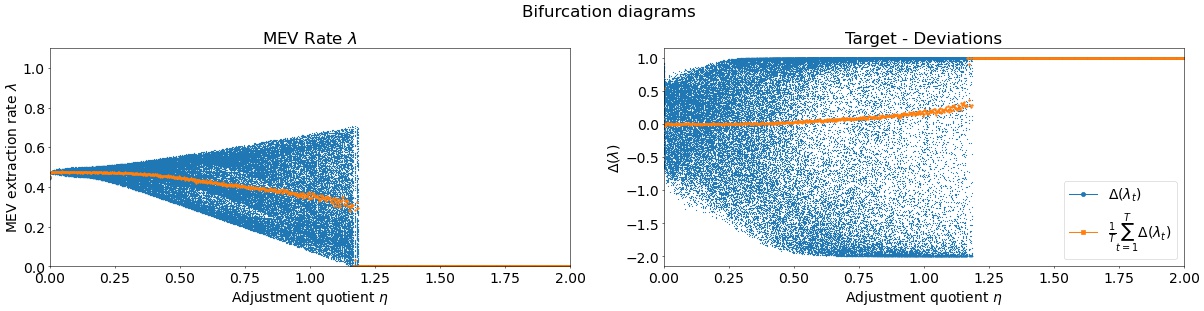}
    \caption{Evolution of the \eqref{eq:mev_dynamics} when a portion of the MEV is burnded. In this simulation, block producers only receive $k$ of their allocated share, where $k$ is parameter that is randomly drawn at each iteration from $[0.8, 0.95]$. Qualitatively equivalent results obtain for different ranges or deterministic values of $k$.}
    \label{fig:burn}
\end{figure}

\subsection{Stress testing under volatile market conditions}
\label{sec:robustness_app}
In this section, we provide further simulations to test the robustness of our model. We remind that in the theoretical part, for our analysis, we considered stationary tolerance distributions for users and miners. However, a shift in the prevailing market conditions, dynamics, or dominant trends is a well studied phenomenon in finance and it is called regime changes \cite{Ang12}. In financial markets, a regime change refers to a transition from one market regime to another, characterized by distinct behavior and patterns. Regime changes can have a profound impact on various aspects of financial markets, including asset prices, volatility, investor sentiment, and trading strategies.\par

\begin{figure*}[t]
    \centering
    \includegraphics[width=0.96\linewidth]{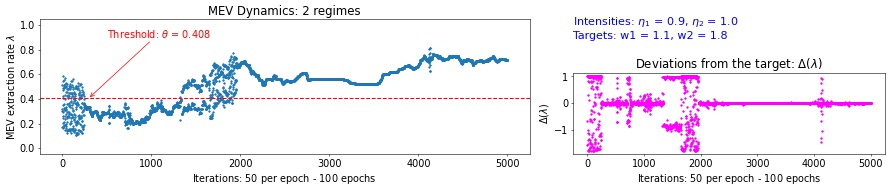}
    \caption{In this scenario, we consider a two-regime market with a threshold $\theta=0.408$ on the extraction rate governing the transition. We handle two control parameters based on the current extraction rate. The left panel shows the evolution of \eqref{eq:mev_dynamics} for 5,000 iterations, which is divided in groups of 100 epochs. The right panel shows the deviations from the target $\Delta(\la)$ during those 100 epochs.}
    \label{fig:regimes}
\end{figure*}

\paragraph*{Regime changes}
In \Cref{fig:regimes}, we examine a realistic scenario with two market regimes. We envision an environment where there are two pairs of tolerance distributions $F_1, G_1$ and $F_2, G_2$. The transition between these two regimes is governed by a threshold of the MEV extraction rate (set at $\theta=0.408$ in the experiments). When the MEV extraction rate is below (above) $\theta$ users and miners are in regime 1 (2). Building on this idea, we also adopt two different intensity values, $\eta_1, \eta_2$, for the updates based on the MEV extraction rate, and two target values $w_1$ and $w_2$. It is known that MEV extracted differs from one application to another. For instance, the MEV generated in decentralized exchanges is regular, whereas liquidations and NFT generate spikes of MEV. Thus, when a spike in MEV is observed, the updates must be more aggressive in order to restore the equilibrium or return to the bounded regions faster. 
In this scenario, we observe the evolution of \eqref{eq:mev_dynamics} for 5,000 iterations. Initially the control parameter is set at $\eta_1=0.9$ and the target at $w_1=1.1$. If at a given time $t$ the extraction rate becomes larger than the threshold, i.e., $\la_t > \theta$, this triggers a change in the control parameter to $\eta_2=1.0$ to deal with the new regime. Vice versa, when the \eqref{eq:mev_dynamics} undershoots the threshold, the control parameter sets back to $\eta_1$. 

We observe here that the extraction rate stabilizes and in fact it stabilizes within one of the two regimes -- see \Cref{fig:regimes} (left). We also see quite good stabilisation of the target values in the right part of the figure. This experiment suggests that our model shows reasonable stability even when the market conditions change between two regimes.   

\begin{figure*}[pt]
    \centering
    \includegraphics[width=0.96\linewidth]{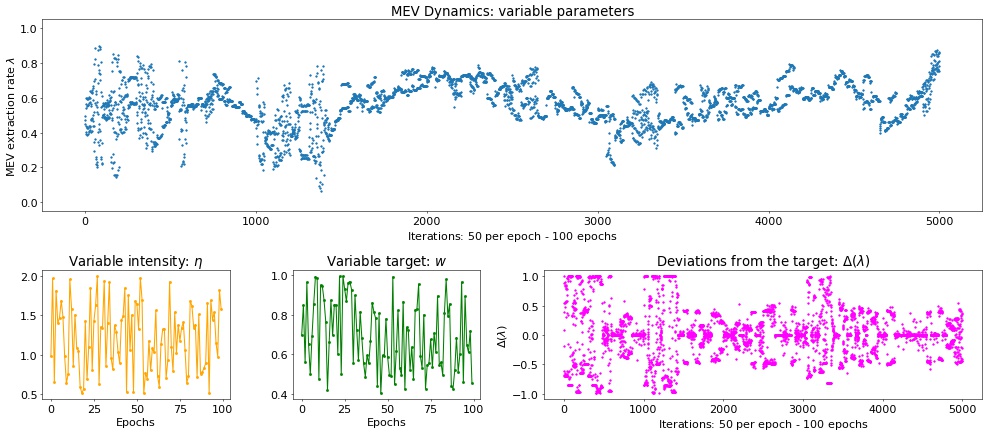}
    \caption{In this scenario, we consider an adversarial setting where the intensity of the updates and the target is randomly sampled in each epoch. The upper panel shows the evolution of \eqref{eq:mev_dynamics} in 5000 iterations, that is divided in 100 epochs. The bottom left panel shows the control parameters randomly drawn over the 100 epochs. The middle panel shows the target values drawn over the 100 epochs. The right panel shows the deviations $\Delta(\la)$ from the target over the 5000 iterations}
    \label{fig:variable}
\end{figure*}

\paragraph*{Unstable market conditions}
In the second scenario, \Cref{fig:variable}, we consider an adversarial scenario akin to stress tests used by regulators and financial institutions in classical financial markets. In this setting, we split the evolution of \eqref{eq:mev_dynamics} in 100 epochs, where each epoch consists of 50 consecutive blocks. In each epoch, we sample randomly a control parameter $\eta$ from the interval $[0.4,1]$ and target $w$ from the interval $[0.5,2]$. This simulates very drastic changes in the parameter $\eta$ and $w$ -- see the yellow and green graphs in \Cref{fig:variable}. Moreover, every 10 blocks, we also perturb the tolerance (beta) distributions $F,G$ by changing their respective parameters $a_u,b_u$ and $a_m,b_m$. This perturbation emulates the (high) volatility in demand that may characterise financial markets in certain periods.\par
It is worth noting, that changes in $\eta$ can be also interpreted as reflecting changes in the generated MEV. As mentioned above, if the generated MEV in period $t$ is $m_t$, then the dynamics become 
 \[\lambda_{t+1}=\lambda_t+\underbrace{\eta m_t}_{:=\eta_t}\lambda(1-\lambda_t)\left(\bar{F}(\lambda_t)-wG(\lambda_t)\right).\]
This generates a sequence $\eta_t:=(\eta m_t)_{t\in \mathbb{N}}$ of adjustment parameters which effectively encodes the changing value of MEV at each update period (block or epoch).\par
Even in this extreme stress test scenario, the results support our theoretical findings. The market remains in a safe zone, in the sense that the system keeps fluctuating and is bounded around the respective target -- see the top blue chart in \Cref{fig:variable}. As it is shown in \Cref{lem:attracting}, the bounded area expands as the control parameter $\eta$ increases. This is evident also in this simulation if we compare the evolution of \eqref{eq:mev_dynamics} with the trajectory of the random values drawn of the control parameter $\eta$. The evolution of the target value (bottom right pink diagram in \Cref{fig:variable}) exhibits a more challenging behavior, which on one hand shows how demanding our stress test is whilst on the other supports the need to investigate more dynamic environments. \par


\section{Discussion}
\label{sec:discussion}

The previous mathematical analysis and simulations demonstrate the properties and performance guarantees of dynamic MEV sharing through extraction rate oracles. In this section, we discuss the related proposals and implementations and also the economic considerations that emerge during the implementation of this mechanism.
\paragraph*{Connection to known protocols}
Our work is related to the recent proposal of MEV-burn \cite{Drake23burn}, which is an add-on to an enshrined PBS framework. We remind that burning is implicitly a redistribution to the token holders. In the aforementioned proposal, in each block, a committee of \textit{attesters} agree on a lower bound of the observed MEV, which acts as a reserve price for the builders' auction. The current proposer can only accepts bids higher than the agreed lower bound. Eventually, the reserve price of the bid is burnt and the rest is kept by the proposer. Thus, the relation of this proposal with our model is straightforward, in the sense that, in each block, the redistributed amount (to the token holders) is  
\begin{equation*}
\text{reserve price} = (1- \la) \cdot \text{winning bid}
\end{equation*}
In our model though, the redistributed amount is not solely defined by the observed MEV, but it is also governed by the feedback control \eqref{eq:mev_dynamics}. We conjecture that our mechanism further smooths out the MEV spikes, yielding security benefits. Such a comparative analysis is left for future work, since the main goal of this work is the analysis of the dynamic evolution of \eqref{eq:mev_dynamics}. \par

Our work is also closely related to the recent product of Flashbots, MEV-Share \cite{Flashbots_share}. In this protocol, users can customize the visibility of their transaction data to searchers, but most notably specify the distribution of the searchers payment \cite{Flashbots_share_docs_refund}. In particular, the user specifies the amount of the searcher's payment to be refunded (i.e, $1-\lambda$), while the rest of it (i.e., $\lambda$) goes to the validator. If not specified, the user will by default receive $90\%$ of the searcher's payment, and the remaining goes to the validator. The inherent tradeoff that user has to consider is that by keeping a larger refund, this may result in longer block inclusion times since it reduces the payment to the validator. Thus, the user is, implicitly, searching for the optimal refund that balances her utility and validator's payoff. In contrast, our proposed mechanism seeks to balance users' and block producers' utilities \textit{internally} and thus it act as a posted-price mechanism, meaning that it offers each user a take-it-or-leave-it refund. This reduces the cognitive burden for users and hence improves the user experience. We further elaborate on the user experience in the next paragraph.

\paragraph*{Incentives: user experience and market stability}
The main goal of dynamic MEV sharing through commonly known extraction rates is to reduce or eliminate uncertainty about the actual MEV extraction that users (respectively, block producers) will be facing (respectively, enjoying). The extraction rates maintained and broadcast by the mechanism provide an oracle to both users and block producers that they can consult in advance of participating in the specific market. This enhances both user experience, and market stability. The ideal extraction rate, which may dynamically change over time in response to changing market conditions, ensures that incentives are optimally balanced between users and block producer. In turn, this ensures sufficient transaction throughput on the one side and participation in the block-production process on the other side to safeguard the execution of transactions. \par

\paragraph*{Selecting system parameters}
The implementation of the dynamic MEV extraction rate mechanism requires the selection of a target ratio and adjustment quotient. To optimally select these parameters in practice, market designers can use a wide range of measurable signals, including historical data, competition with peers and target performance. However, as shown by our analysis, while the exact evolution of these systems is highly dependent on the selection of their parameters, their performance guarantees, i.e., market liveness and bounded deviations, are robust across a wide range of market conditions and mechanism designs.

\paragraph*{Strategic user behaviour}
The implementation of an MEV extraction oracle inevitably creates new vectors for strategic behavior among users and block producers. Both parties might attempt to influence the mechanism's future states through their current actions, such as abstaining from or delaying transaction submissions or executions to manipulate the extraction rate in their favor. While significant, the implicit effects of such strategic behaviors are complex to analyze due to interwined higher-order effects. For example, if many users delay their transactions during periods of high extraction rates, this can result in waiting costs (costs from delayed execution) and create congestion, leading to further delays and higher execution costs due to increased transaction fees. Similar considerations apply for block builders. Thus, although such behaviors are also possible in transaction fee mechanisms like EIP-1559, they have rarely been observed in practice, likely due to the complexity involved \cite{Rou21, Leonardos2022, Ba23}. Consequently, while static or dynamic attack vectors and strategic behavior can significantly impact any such mechanism, their analysis remains a separate and underexplored topic of its own interest in the literature.

\section{Conclusion}
\label{sec:conclusion}
Maximal Extractable Value (MEV) constitutes a critical issue in smart-contract blockchains and DeFi markets. With the evolution of the MEV ecosystem, the whole process of extracting excess value from the economic activity of users has been systematized and benefits different kind of blockchain participants including searchers, builders, validators and block proposers. This creates a fundamental trade-off between users' sovereignty over their own transactions and revenue maximization for the entities who secure and process these transactions. Accordingly, a lot of research is concentrated on designing mechanisms that will return to users part of the excess value created by their economic activity.\par
Our work establishes the first to our knowledge dynamic framework for the integral problem of MEV sharing. Without arguing on whether MEV is desirable or not, this framework treats MEV as an unavoidable phenomenon of current blockchains and enables the market to decide dynamically on its exact level through a protocol-based implementation. Using the design principles of the EIP-1559 transaction fee mechanism, our rigorous mathematical results provide the first formal evidence that similar mechanisms, with properly engineered target functions, can be successful in tackling broader blockchain optimization problems. From a practical perspective, our results suggest that dynamic MEV extraction can achieve targets set by the market designers and can, thus, inform the ongoing discussion about actively enshrining MEV extraction in DeFi protocols. 

\section*{Acknowledgments}
This research was supported in part by the National Research Foundation, Singapore and DSO National Laboratories under its AI Singapore Program (AISG Award No: AISG2-RP-2020-016), grant PIESGP-AI-2020-01, AME Programmatic Fund (Grant No.A20H6b0151) from A*STAR.


\bibliographystyle{alpha}
\bibliography{v2}

\appendix
\section{Omitted Proofs}
\label{appendix}

\begin{proof}[Proof of \Cref{thm:performance}]
Since $\Delta(\la_t)$ is decreasing, we have by \Cref{lem:attracting} that 
\[\Delta\left(\la^*-\frac{w\eta}{4}\right)\ge \Delta(\la_t)\ge \Delta\left(\la^*+\frac{\eta}{4}\right),\]
for all $\la_t$'s that lie in the attracting region of the dynamics, i.e., in $[\la^*-\frac{w\eta}{4},\la^*+\frac{\eta}{4}]$ provided that $\eta$ is such that $0<\la^*-\frac{w\eta}{4}, \la^*+\frac{\eta}{4}<1$. This implies that 
\begin{align*}
|\Delta(\la_t)| &\leq \max \left\{
 \left|\Delta\left(\la^*-\frac{w\eta}{4}\right)\right| , \left|
\Delta\left(\la^*+\frac{\eta}{4}\right)\right| \right\} \leq \left|\Delta\left(\la^*-\frac{w\eta}{4}\right)\right| + \left|\Delta\left(\la^*+\frac{\eta}{4}\right)\right| \, .
\end{align*}
Using the Mean Value Theorem, and the fact that $\Delta(\la^*)=0$, we have that \[\Delta\left(\la^*-\frac{w\eta}{4}\right)= \Delta\left(\la^*-\frac{w\eta}{4}\right) - \Delta(\la^*)=-\Delta'(x_1)\cdot\frac{w\eta}{4},\] for some $x_1\in \left(\la^*-\frac{w\eta}{4},\la^*\right)$,
and similarly that $\Delta\left(\la^*+\frac{\eta}{4}\right)=\Delta\left(\la^*+\frac{\eta}{4}\right)-\Delta(\la^*)=\Delta'(x_2)\cdot\frac{\eta}{4}$, for some $x_2\in \left(\la^*,\la^*+\frac{\eta}4\right)$. Thus, to obtain meaningful bounds for $\Delta(\la_t)$, it suffices to bound the derivative of the function $\Delta(\la_t)$ for $\la_t$ in the specified regions around $\la^*$. It holds that
    \begin{align*}
        |\Delta'(x)|&=|\bar{F}'(x)-wG'(x)|=|-f(x)-wg(x)|
        \\&\le \frac{1}{B(a,b)}x^{a-1}(1-x)^{b-1}+\frac{w}{B(c,d)}x^{c-1}(1-x)^{d-1} = H(x).
    \end{align*}
Note that all coefficients in this expression are positive. Thus, we can complete the proof since
\begin{align*}
 |\Delta(\la_t)| &\leq \left|\Delta\left(\la^*-\frac{w\eta}{4}\right)\right| + \left|\Delta\left(\la^*+\frac{\eta}{4}\right)\right| \leq
 \frac{\eta w}{4} \cdot \max_{x \in (p,\lambda^*)} \{H(x)\}  +
 \frac{\eta}{4} \cdot \max_{x \in (\lambda^*,q)} \{H(x)\}    
 \\&\leq \frac{\eta (1+w)}{4} \cdot \max_{x \in (p,q)} \{H(x)\},
\end{align*}
which is less or equal than 
\begin{align*}
    \frac{\eta (1+w)}{4} \cdot &\left(\frac{1}{B(a,b)} \cdot \max_{x \in (p,q)} \{x^{a-1}(1-x)^{b-1}\}+\frac{w}{B(c,d)} \cdot \max_{x \in (p,q)} \{x^{c-1}(1-x)^{d-1}\}\right),
\end{align*} as claimed.
\end{proof}

\begin{proof}[Proof of \Cref{lem:beta}]
  Let us first observe that $f(x) > 0$ for all $x \in [p,q]$. We have that
  \begin{align*}
    f'(x) &= (a-1)x^{a-2}(1-x)^{b-1} - x^{a-1}(b-1)(1-x)^{b-2} \\&= 
    x^{a-2}(1-x)^{b-2} ( (2-a-b)x +a -1).
  \end{align*}
  We will now inspect the behavior of function $f$ for $x \in [p,q] \subset (0,1)$. 
  Then, because $x^{a-2}(1-x)^{b-2} > 0$, unless $x\in\{0,1\}$, the monotonicity of $f$ is determined
  by the sign of the linear function $(2-a-b)x + a-1$ of $x$.

  Suppose first that $a+b \not = 2$, then $f'(\zeta) = 0$ for $\zeta = \frac{1-a}{2-a-b}$. If $a+b < 2$, then the function $(2-a-b)x + a-1$ of $x$ is strictly increasing and its zero is $\zeta$. This implies that $(2-a-b)x + a-1 < 0$ for $x < \zeta$ and $(2-a-b)x + a-1 > 0$ for $x > \zeta$. So function $f$ is strictly decreasing when $x < \zeta$ and strictly increasing when $x > \zeta$, achieving its minimum at $x=\zeta$.

  We now consider the following cases:
  \begin{itemize}[leftmargin=*]
      \item $\zeta < p$: in this case function $f$ in increasing when $x \in [p,q]$, meaning that it reaches its maximum at $x=q$         
      \item $p \leq \zeta \leq q$: in this case function $f$ in decreasing when $x \in [p,\zeta]$, and increasing when $x \in [\zeta,q]$, meaning that it reaches its maximum at $x=p$ or $x=q$, depending which of the two values $f(p)$, $f(q)$ is larger.
      \item $\zeta > q$: in this case function $f$ in decreasing when $x \in [p,q]$, meaning that it reaches its maximum at $x=p$.
  \end{itemize}

  If $a+b > 2$, then the function $(2-a-b)x + a-1$ of $x$ is strictly decreasing, that is, $(2-a-b)x + a-1 > 0$ for $x < \zeta$ and $(2-a-b)x + a-1 < 0$ for $x > \zeta$. Thus, function $f$ is strictly increasing when $x < \zeta$ and strictly decreasing when $x < \zeta$, achieving its maximum at $x=\zeta$.

  We now consider the following cases:
  \begin{itemize}[leftmargin=*]
      \item $\zeta < p$: in this case function $f$ in decreasing when $x \in [p,q]$, meaning that it reaches its maximum at $x=p$.      
      \item $p \leq \zeta \leq q$: in this case function $f$ in increasing when $x \in [p,\zeta]$, and decreasing when $x \in [\zeta,q]$, meaning that it reaches its maximum at $x=\zeta$.
      \item $\zeta > q$: in this case function $f$ in increasing when $x \in [p,q]$, meaning that it reaches its maximum at $x=q$.
  \end{itemize}
Let us finally consider the case when $a+b=2$. Then $f'(x) = x^{a-2}(1-x)^{b-2} (a-1) > 0,$ when $a>1$ and $f'(x) < 0$ when $a<1$. In the former case, function $f$ is strictly increasing for $x \in [p,q]$ and therefore reaches it maximum at $x=q$. In the latter case, function $f$ is strictly decreasing for $x \in [p,q]$ and therefore reaches it maximum at $x=p$. Finally, if $a+b=2$ and $a=1$, then $b=1$, therefore function $f$ is a constant function $f(x) = 1$ for all $x$.
\end{proof}

\end{document}